\documentclass[reqno,11pt]{amsart}

\usepackage{amsmath,amssymb,amsthm}
\usepackage[usenames,dvipsnames]{color}
\usepackage[bookmarks,colorlinks,breaklinks]{hyperref}  % PDF hyperlinks, with coloured links
\usepackage{doi,url}
\usepackage{graphicx}
\usepackage{enumerate}
\usepackage{bbm}
\usepackage{mathabx}
\usepackage{appendix}
\usepackage{tikz}

\definecolor{dullmagenta}{rgb}{0.4,0,0.4}   % #660066
\definecolor{darkblue}{rgb}{0,0,0.4}
\hypersetup{linkcolor=red,citecolor=magenta,filecolor=dullmagenta,urlcolor=darkblue}
\definecolor{trolleygrey}{rgb}{0.5, 0.5, 0.5}
\numberwithin{equation}{section}

\newtheoremstyle{ttheorem}%
       {1.3ex\@plus1ex}                % space above
       {2.5ex\@plus1ex\@minus.5ex}      % space belowi
       {\itshape}           % body font
       {0pt}                   % indent amount
       {\bfseries}          % Theoremhead font
       {.}                  % Punctuation after theorem head
       {.5em}               % Space after theorem head
       {}                % Theorem head spec (can be left empty: `normal')

\newtheoremstyle{ddefinition}%
       {1.3ex\@plus1ex}                % space above
       {2.5ex\@plus1ex\@minus.5ex}      % space belowi
       {}           % body font
       {0pt}                   % indent amount
       {\bfseries}           % Theoremhead font
       {.}                  % Punctuation after theorem head
       {.5em}               % Space after theorem head
       {}                % Theorem head spec (can be left empty: `normal')

\newtheoremstyle{rremark}%
       {1.3ex\@plus1ex}                % space above
       {2.5ex\@plus1ex\@minus.5ex}      % space belowi
       {\itshape}        % body font
       {0pt}                   % indent amount
       {\bfseries}           % Theoremhead font
       {.}                  % Punctuation after theorem head
       {.5em}               % Space after theorem head
       {}                   % Theorem head spec (can be left empty: `normal')

\theoremstyle{ttheorem}
\newtheorem{theorem}{Theorem}[section]

\newtheorem{rem}[theorem]{Remark}

\theoremstyle{ddefinition}

    {\end{nummer}\end{myremarks}}

\newcounter{numcount}
\newcommand{\labelnummer}{\mbox{\normalfont (\roman{numcount})}}%

\makeatletter

\newenvironment{nummer}%
  {\let\curlabelspeicher\@currentlabel%
    \begin{list}{\labelnummer}%
      {\usecounter{numcount}\leftmargin0pt%
        \topsep0.5ex\partopsep2ex\parsep0pt\itemsep0ex\@plus1\p@%
        \labelwidth2.5em\itemindent3.5em\labelsep1em%
      }%
    \let\saveitem\item%
    \def\item{\saveitem%
      \def\@currentlabel{\curlabelspeicher$\,$\labelnummer}}%
    \let\savelabel\label%
    \def\label##1{\savelabel{##1}%
      \@bsphack%
        \ifmmode\else%
          \protected@write\@auxout{}%
          {\string\newlabel{##1item}{{\labelnummer}{\thepage}}}%
        \fi%
      \@esphack%
    }%
  }{\end{list}}%

\def\itemref#1{\expandafter\@setref\csname r@#1item\endcsname%
  \@firstoftwo{#1}}%

\def\section{\@startsection{section}{1}%
  \z@{1.3\linespacing\@plus\linespacing}{.5\linespacing}%
  {\normalfont\scshape\centering}}

\renewcommand\L{\mathrm{L}}

\newcommand\cI{\mathcal{I}}

\newcommand\RR{\mathbb{R}}
\newcommand\NN{\mathbb{N}}
\newcommand\ZZ{\mathbb{Z}}

\newcommand\CC{\mathbb{C}}

\newcommand\eps{\varepsilon}

\renewcommand\P{\mathbb P}
\newcommand\E{\mathbb E}

\def\Chi{\raisebox{.4ex}{$\chi$}}

 \let\Im\undefined

\DeclareMathOperator{\Im}{Im}
\DeclareMathOperator{\dom}{dom}
\DeclareMathOperator{\supp}{supp}
\DeclareMathOperator{\dist}{dist}

\def\le{\leqslant}

\usepackage[active]{srcltx}
\theoremstyle{ttheorem}
 \newtheorem{thm}{Theorem}[section]
 
 \newtheorem{cor}[thm]{Corollary}
 \newtheorem{lem}[thm]{Lemma}
 \newtheorem{prop}[thm]{Proposition}
 \theoremstyle{definition}
 \newtheorem{defn}[thm]{Definition}

 \theoremstyle{remark}

 \numberwithin{equation}{section}

\newcommand{\be}{\begin{equation}}
\newcommand{\ee}{\end{equation}}
\newcommand{\benon}{\begin{equation*}}
\newcommand{\eenon}{\end{equation*}}
\newcommand{\ba}{\begin{array}}
\newcommand{\ea}{\end{array}}
\newcommand{\bal}{\begin{align}}
\newcommand{\eal}{\end{align}}
\newcommand{\bea}{\begin{eqnarray}}
\newcommand{\eea}{\end{eqnarray}}
\newcommand{\bee}{\begin{eqnarray*}}
\newcommand{\eee}{\end{eqnarray*}}

\newcommand{\norm}[1]{\Vert #1 \Vert}
\newcommand{\abs}[1]{\left| #1 \right|}

\newcommand{\Lp}[1]{\textrm{L}^2(#1)}
\renewcommand{\L}{\Lambda}

\newcommand{\angles}[1]{\langle #1 \rangle}

\renewcommand{\dom}{{D^\omega}}
\newcommand{\om}{\omega}

\newcommand{\pa}[1]{\left( {#1} \right)}

\newcommand{\hm}[1]{\leavevmode{\marginpar{\tiny%
$\hbox to 0mm{\hspace*{-0.5mm}$\leftarrow$\hss}%
\vcenter{\vrule depth 0.1mm height 0.1mm width \the\marginparwidth}%
\hbox to
0mm{\hss$\rightarrow$\hspace*{-0.5mm}}$\\\relax\raggedright #1}}}

\begin{document}

\title[Random operators in discrete structures]{Random Schr\"odinger Operators on discrete structures}

 %\author[F. Hoecker-Escuti]{F. Hoecker-Escuti}
  \author[C. Rojas-Molina]{C. Rojas-Molina}
%\address[C. Rojas-Molina]{}
% \email{crojasm@uni-bonn.de}
%\thanks{Version of \today}
\maketitle
\begin{abstract}
The Anderson model serves to study the absence of wave propagation in a medium in the presence of impurities, and is one of the most studied examples in the theory of quantum disordered systems. In these notes we give a review of the spectral and dynamical properties of the Anderson Model on discrete structures, like the $d$-dimensional square lattice and the Bethe lattice, and the methods used to prove localization. These notes are based on a course given at the CIMPA School ''Spectral Theory of Graphs and Manifolds'' in Kairouan, 2016.
\end{abstract}

\section{Introduction}

The Anderson model was proposed by P.W. Anderson in his ground-breaking article from 1958 \cite{And58} to explain the absence of diffusion of quantum waves in disordered lattices. This phenomenon is known today as Anderson localization, and earned his discoverer the Nobel Prize in physics in 1977. Since the late 70s, the mathematical-physics community has invested many efforts in obtaining a rigorous description of this phenomenon, and today, despite great progress, many of the original questions remain unsolved. In order to study the propagations of electrons in a solid, we work in the framework of quantum mechanics. An electron moving in a space $\Gamma$ at a given time is described by a normalized wave function $\psi$ in a suitable Hilbert space $\mathcal H$, whose evolution in time is given by the Schr\"odinger equation:
\[i\partial_t \psi(x,t)=(-\Delta+V)\psi(x,t),\quad x\in\Gamma,\,t\in\mathbb R.\]
Here, the one-particle Schr\"odinger operator $H:=-\Delta+V$ is a self-adjoint linear operator acting on $\mathcal H$ that  represents the energy of the particle. Namely, the negative Laplacian $-\Delta$ represents the kinetic energy and the potential $V$ encodes the interaction between the particle and the atomic structure of the solid. Therefore, if the initial state of an electron is described by $\psi(x,0)$, its state at a time $t$ is given by
\[\psi(x,t)=e^{-itH}\psi(x,0),\]
where the right hand side is well defined using the spectral theorem for self-adjoint operators. Knowing the spectral and dynamical properties of the operator $H$ yields information on how the electron propagates in the material in time. %A solid can behave as a conductor or an insulator, depending if the electrons can propagate among
% the orbits around the nuclei in the atomic structure of the solid.

In terms of the dynamics of the particle, if the electron propagates (corresponding to a conducting behavior of the material), the associated wave function is \emph{extended} (ex. $e^{ix}$). On the contrary, if the electron does not propagate, the wave function is \emph{localized} (ex. $e^{-x^2}$), in which case the material behaves as an insulator.
An analogue duality can be seen in the spectral theoretical decomposition of the spectrum of the operator $H$. The spectrum of $H$ can be decomposed into a pure-point and a continuous part \cite{RSI}. The existence of pure-point spectrum is called spectral localization. Despite what the names might suggest, there is no exact equivalence between the spectral type of $H$ and the evolution of $\psi(x,t)$, as we will see later in Section \ref{s:loc}.

P.W. Anderson observed that the presence of impurities in the environment, coming from either the composition of the atoms or the space distribution of the nuclei in the atomic structure, was, under certain conditions, enough to suppress the propagation of electrons, turning the material into an insulator.  To explain this phenomenon, Anderson proposed to study a Schr\"odinger operator where the impurities are encoded in the potential in the form of realizations of a random variable in some suitable probability space $\Omega$. In this way, the Anderson model is a random Schr\"odinger operator
\be H_\omega=-\Delta+\lambda V_\omega \quad \mbox{on }\ell^2(\ZZ^d),\ee
 where $-\Delta$ is the negative discrete Laplacian, $V_\omega$ with $\omega\in\Omega$ is a random operator and $\lambda>0$ is a real parameter representing the strength of the disorder. The localization phenomenon observed for $H_\omega$ is exclusively caused by the randomness in $V_\omega$. This is to be differentiated from other types of localization caused by, for exemple, interactions (Mott localization).

\begin{figure}[h]
  \centering
  \includegraphics[width=9cm]{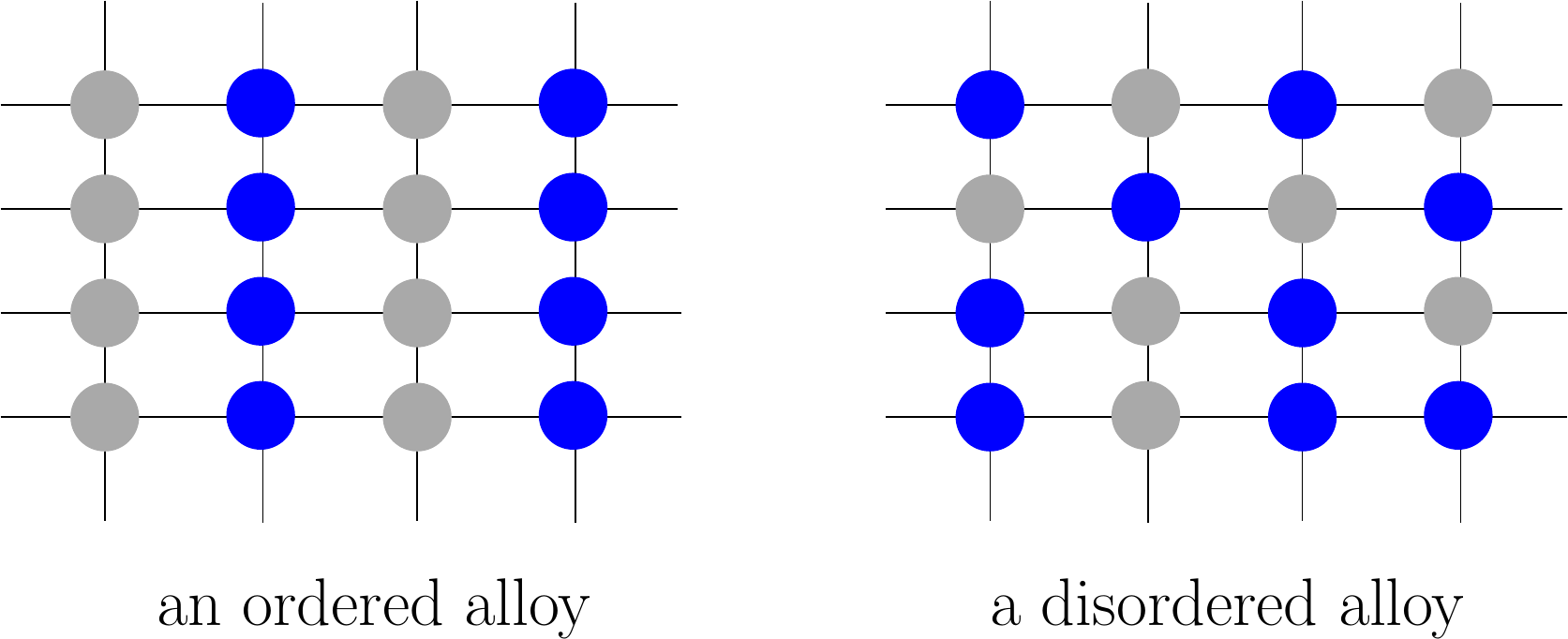}\\
  \caption{A depiction of potentials in an ordered and disordered material. In this case an alloy, where impurities correspond to the species of atom sitting on each point of the lattice. Both pictures depict a possible configuration $\omega\in\Omega$, the left one depicts a periodic configuration and therefore the resulting $V_\omega$ is periodic}
\end{figure}

 The localization properties of the Anderson model have been extensively studied in the mathematics literature since the late 70s, starting with the work of Gold'sheid, Molchanov and Pastur \cite{GMP77} and Kunz-Souillard \cite{KuSo80}. In one dimension, the operator typically exhibits localization in the whole spectrum irrespective of the intensity of the disorder. In two dimensions and higher, the operator typically exhibits \emph{localization} in the whole spectrum at high disorder, or at spectral band edges, if the disorder strength is moderate. In the particular case of two dimensions, the Anderson operator is expected to exhibit localization throughout the spectrum at any disorder strength, as in the one-dimensional case. The proof of this remains, however, an open problem. In dimensions three and above, the operator is expected to undergo a transition from extended to localized states, known as the \emph{Anderson metal-insulator transition}, exhibiting \emph{localization} at spectral band edges and \emph{delocalization} in the bulk of the spectrum. This can be understood as different regimes where one of the two parts of the operator dominates the picture: either the free part $-\Delta$ dominates and imposes its absolutely continuous spectrum with associated extended states, or the random perturbation $V_\omega$ dominates, with its pure point spectrum and associated localized eigenfunctions.

\begin{figure}[h]
  \centering
  \includegraphics[width=9cm]{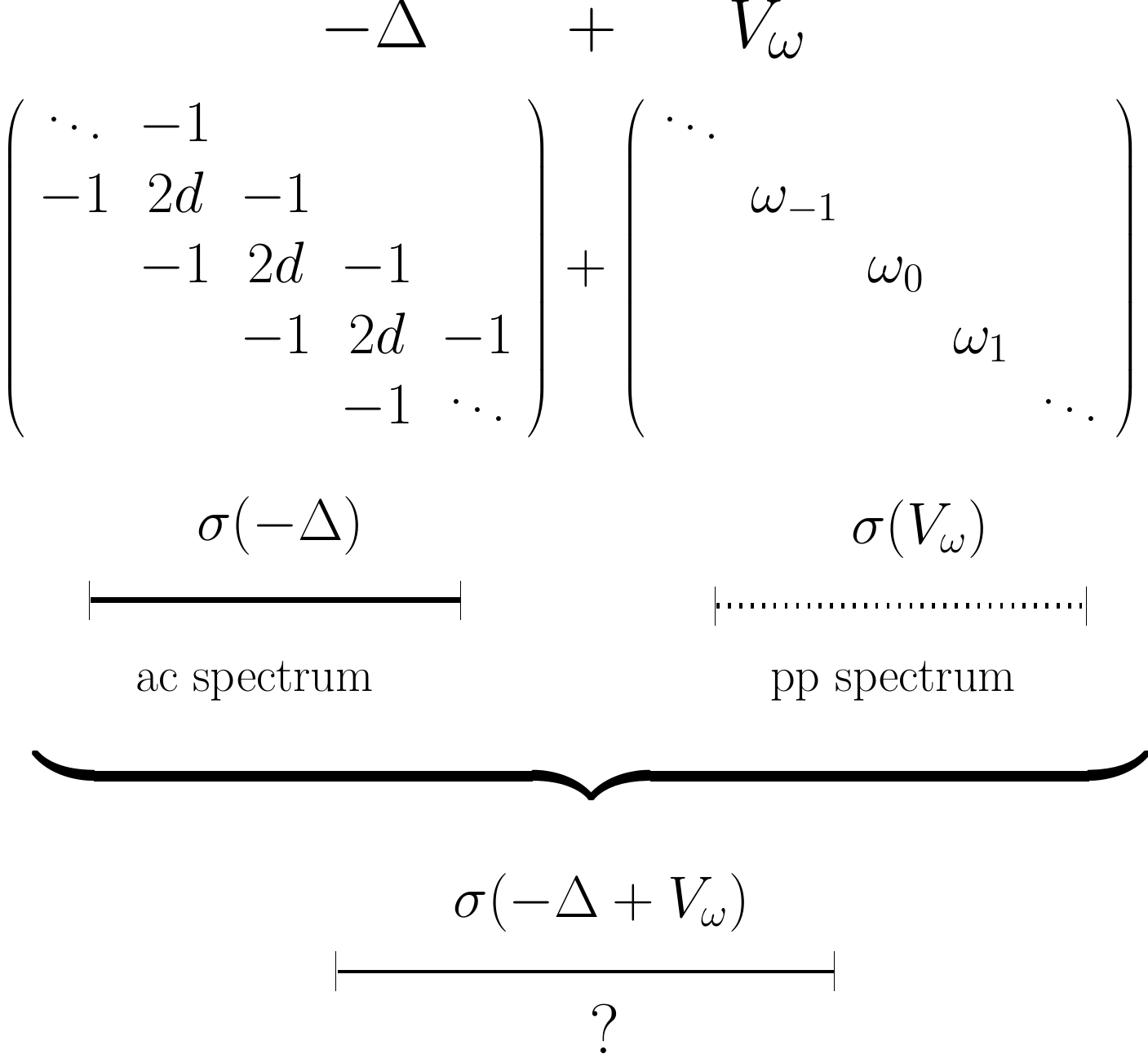}\\
  \caption{Here, $\sigma(A)$ denotes the spectrum of an operator $A$}
 \label{sumspectra}
\end{figure}

\bigskip

  Delocalization for the Anderson model, that is, the appearance of extended states, in the bulk of the spectrum has only been proven on the Bethe lattice $\mathbb B$, also called Cayley tree \cite{Kl94, Kl98, ASW06,FHS07,AW13}, and a proof of delocalization for the Anderson model on the square lattice $\mathbb Z^d$ remains the biggest open problem in the field.

  In arbitrary dimensions two methods are available to prove localization: the Multiscale Analysis, developed by J. Fr\"ohlich and T. Spencer in the early 80s \cite{FrS83}, and the Fractional Moment Method, developed by M. Aizenman and S. Molchanov ten years later \cite{AM93}. Streamlined early versions of the Multiscale Analysis can be found in the books by Pastur and Figotin \cite[Section 15.C]{PF92}, and Carmona and Lacroix \cite[Chapter IX]{CaLa90}. These textbooks contain a comprehensive account of the spectral theory for general random Schr\"odinger operators, including the Anderson model on $\mathbb Z^d$ and the one-dimensional case, where methods from dynamical systems can be used to prove localization.
  Today, both monographs are standard references in the theory of disordered quantum systems. Other early references of interest, more physics-oriented, are \cite{MS87,Sp,M90}. Among the specialized references for the two existing methods to prove localization in arbitrary dimension, the monograph by P. Stollmann \cite{Sto} focuses on the Multiscale Analysis, while the more recent work by M. Aizenman and S. Warzel \cite{AW} explains in detail the Fractional Moment Method and its connection with statistical mechanics. An object of great importance in the study of disordered materials is the integrated density of states, which is also fundamental in the proofs localization and of interest in itself, see e.g. the monograph by I. Veselic \cite{V}.

   The busy reader who does not feel the need to know every result in the field can find in the literature excellent short monographs concerning the Anderson model and its localization properties: the lectures notes by W. Kirsch \cite{Kirsch89, K} and the more advanced by A. Klein \cite{Kl} put emphasis on the Multiscale Analysis, while those by G. Stolz \cite{S10} focus on the Fractional Moment Method, as do the notes by D. Hundertmark \cite{H08}, who gives a more probabilistic approach to the method. We encourage the interested reader to look at these references to have a first understanding of the proofs of localization in random Schr\"odinger operators. In the present article, we build on the aforementioned works and give a complementary view focusing on the Anderson model on the discrete setting.  We try to avoid overlaps, while at the same time trying to be as self-contained as possible. At times we will inevitably fail at one task or the other.

In Section \ref{s:and} we set the model in a general framework that allows us to consider both the $d$-dimensional square lattice $\mathbb Z^d$ and the Bethe lattice $\mathbb B$. In Section \ref{s:loc} we discuss the various notions of localization and how they are related. In Section \ref{s:dgf} we discuss the methods known to prove localization in arbitrary dimension, while in Section \ref{s:deloc} we comment on delocalization on the Bethe lattice $\mathbb B$.
 %In addition, in {\purple In Section \ref{s:ids} we introduce the Integrated Density of States and give a proof of Lifshitz tails.}

%\vspace{-0.5cm}

\section{The Anderson model on discrete structures}\label{s:and}
\subsection{A brief recall of probability theory}
Consider a probability space $(\Omega,\mathcal B,\mathbb P)$, where $\mathcal B$ is a $\sigma$-algebra on $\Omega$ and $\mathbb P$ is a probability measure on $(\Omega,\mathcal B)$. Given a probability space $(\Omega,\mathcal B,\mathbb P)$, a random variable is a measurable function $X:\Omega\rightarrow \mathbb R$. Sets in $\mathcal B$ are called \emph{measurable sets}. The probability distribution of $X$ is a measure $\P_X$ defined by
 \be \P_X(A)=\P(X\in A):=\P(X^{-1}(A))= \mathbb P(\{\omega\in\Omega;\, X(\omega)\in A \}). \ee
The support of an arbitrary measure $\mu$ is given by
\be \rm{supp}\, \mu:=\{ x\in\mathbb R;\, \mu([ x-\epsilon,x+\epsilon ])>0,\, \forall \epsilon>0 \}.\ee
 If for any $A\in\mathcal B$, $\mathbb P(Y\in A)=\mathbb P(X\in A)=\mu(A)$, we say $X$ and $Y$ are \emph{identically distributed} with common probability distribution $\mu$. We denote by $\E$ the expectation with respect to the probability measure $\P$, that is, $\E(X)=\int_\Omega X d\P(\omega)$.

Given a countable set $\Gamma$, a collection of random variables $(X_i)_{i\in\Gamma}$ defined on the same probability space is called a \emph{stochastic process}. Moreover, the collection $(X_i)_{i\in\Gamma}$ is called \emph{independent} if, for any finite subset $\{n_1,...n_k\}\subset \Gamma$ and arbitrary Borel sets $A_1,...,A_k\subset\mathbb R$,
\be \mathbb P(X_{n_1}(\omega)\in A_1,...,X_{n_k}(\omega)\in A_k)  =\prod_{j=1}^k \mathbb P(X_{n_j}(\omega)\in A_j).\ee
If the collection of random variables $(X_i)_{i\in\Gamma}$ is independent and identically distributed (denoted i.i.d.) with common probability distribution $\mu$, we have
          \be \mathbb P(X_1(\omega)\in A_1,...,X_k(\omega)\in A_k)  = \prod_{j=1}^k \mu(A_j).\ee

Given a sequence of probability spaces $(\mathbb R, \mathcal B_i,\mu_i)$, $i\in\Gamma$, we consider the product probability space $(\Omega,\mathcal B, \P)$, where $\Omega=\bigotimes_{i\in\Gamma}\mathbb R=\RR^\Gamma$, $\P$ is the product probability measure
\be \P=\bigotimes_{i\in\Gamma}\mu, \ee
and $\mathcal B$ is the product $\sigma$-algebra generated by the cylinder sets of the form
\be\label{cysets} \prod_{i\in\Gamma} B_i,\ee
with $B_i\in\mathcal B_i$ and $B_i=\mathbb R$ for all except finitely many $i\in\Gamma$. The existence of the infinite product probability space $(\Omega,\mathcal B,\P)$ is ensured by Kolmogorov's extension Theorem \cite[Appendix, Sect. 7]{Du}. Note that cylinder sets of the form \eqref{cysets} also generate the topology of $\Omega$, induced by that of $\mathbb R$.
We write $\omega:=(\omega_i)_{i\in\mathbb Z^d}$ instead of $\{X_i(\omega)\}_{i\in\Gamma}$.

Given a probability space $(\Omega,\mathcal B,\P)$, we will often be interested in results that hold for $\P$-almost every $\omega\in\Omega$, that is, for $\omega$ in a set $\Omega_0$ with $\P(\Omega_0)=1$. For this, it will be very useful to have the Borel-Cantelli Lemma \cite[Ch.1, Sect. 6]{Du} in our toolbox,
\begin{thm}\label{t:bc}Let $(\Omega,\mathcal B,\P)$ be a probability space and $(B_k)_{k\in\NN}$ a sequence of sets in $\mathcal B$. Define
\be B_\infty=\bigcap_K \bigcup_{k\geq K}B_k. \ee
The following holds,
\begin{itemize}
\item[i.] If $\sum_{k=1}^\infty \P(B_k)<\infty$, then $\P(B_\infty)=0$.
\item[ii.] If the sets $B_k$ are independent and $\sum_k\P(B_k)=\infty$, then $\P(B_\infty)=1$.
\end{itemize}
\end{thm}

\subsection{The Anderson model}
For $\Gamma$ a discrete set, consider the Hilbert space $\mathcal H=\ell^2(\Gamma)$ given by
\[\ell^2(\Gamma)=\{\psi:\Gamma\rightarrow \CC:\,\sum_{x\in\Gamma}\abs{\psi(x)}^2<\infty \},\]
with inner product $\angles{\psi,\phi}=\displaystyle\sum_{x\in\Gamma}\overline{\psi(x)}\phi(x)$ and norm $\norm{\psi}=\sqrt{\angles{\psi,\psi}}$. We write $\norm{\psi}_\infty=\sup_{x\in\Gamma}\abs{\psi(x)}$. We denote by $\ell^2_c(\Gamma)$ the elements of $\ell^2(\Gamma)$ with compact support. We also denote the canonical orthonormal basis of $\ell^2(\Gamma)$ by $(\delta_i)_{i\in\Gamma}$, where $\delta_i$ is given by
\be \delta_i(x) = \begin{cases} 1, & \mbox{if } x=i \\ 0, & \mbox{otherwise }.\end{cases}\ee

In what follows, the underlying space $\Gamma$ will correspond to the set of vertices
of a given graph. By abuse of notation, we will use  the same notation
for the graph $\Gamma$ and its set of vertices $V$ (excepting this paragraph).  We will only
consider undirected graphs $\Gamma=(V,E)$, defined by a set of
vertices $V$ (also called points) and a set of edges
$E \subset V\times V$. Undirected means that if $(x,y) \in E$ then
$(y,x)\in E$. If two points $x,y\in \Gamma$ are connected by an edge,
i.e. $(x,y) \in E$, we will write $x\sim y$ and say that they are
nearest neighbors. The graph $\Gamma$ can be completely described in matrix
form. We define the adjacency matrix $A_\Gamma$ of the graph $\Gamma$
as the matrix with entries
%\[
%A_\Gamma(x,y) := \chi_{E}\left((x,y)\right).
%\]
\be A_\Gamma(x,y) = \begin{cases} 1, & \mbox{if } x\sim y \\ 0, & \mbox{otherwise },\end{cases}\ee
where $A_\Gamma(x,y)=\angles{\delta_x,A_\Gamma\delta_y}$. Note that for an undirected graph this is a symmetric
matrix.

We define the Anderson model acting on $\ell^2(\Gamma)$ by
 \be\label{eq:and} H_{\om,\lambda}\psi=-\Delta\psi+\lambda V_\omega\psi,\ee
with $\lambda>0$. We make the following assumptions:
 \begin{enumerate}
 \item[\bf(A1)] $\Gamma$ is the set of vertices of a locally finite, connected graph, with uniformly bounded vertex degree. That is, every point $x$ in $\Gamma$ has a finite number $\textrm{deg}_\Gamma(x)$ of nearest neighbors and $\sup_{x\in \Gamma} \textrm{deg}_\Gamma(x) \leq K<\infty$. We also assume that the graph has no loops, i.e. $\forall x\in \Gamma,\, x \not \sim x$. With these assumptions, all the entries in the diagonal of $A_\Gamma$, the adjacency matrix of $\Gamma$, vanish and the matrix gives rise to a bounded operator on $\ell^2(\Gamma)$.

     We associate to $\Gamma$  a distance function $\textrm{dist}_\Gamma:\Gamma\times\Gamma\rightarrow [0,\infty)$. If $\Gamma=\ZZ^d$, $\textrm{dist}_\Gamma(x,y)=\norm{x-y}_\infty=\sup_{1\leq j\leq d}\abs{x_j-y_j}$, where $x_j$ denotes the $j$-th component of the vector $x\in\ZZ^d$.

    \item[\bf(A2)] The operator $-\Delta$ is the discrete analog of the negative Laplacian,
     \be\label{eq:lap}-\Delta \psi(x)= -\sum_{y\sim x}\pa{\psi(y)-\psi(x)},\ee
     where the summation index $y\sim x$ runs over the nearest-neighbors  of $x$ in $\Gamma$. We can write $-\Delta=\textrm{deg}_\Gamma-A_\Gamma$, where $A_\Gamma$ is the adjacency matrix, i.e.
     \be\label{eq:adj} A_\Gamma \psi(x)= \sum_{y\sim x}\psi(y),\ee
     and $\bigl( \textrm{deg}_\Gamma \psi \bigr) (x)=\textrm{deg}_\Gamma(x)\psi(x)$ is a multiplication operator. Note that our operator $-\Delta$ correspond to $\Delta$ in Section 4 of \cite{Go}.

     %where the notation $y\sim x$ stands for the nearest-neighbors in $\Gamma$. We can write $-\Delta=\mathcal A-\rm{deg}(\Gamma)$, where $\mathcal A$ is the so-called adjacency matrix
     %\be\label{eq:adj} \mathcal A \psi(x)= -\sum_{y\sim x}\psi(y),\ee
     %and $\rm{deg}(\Gamma)\psi(x)=\rm{deg_x}(\Gamma)$, see \cite[Section 4]{Go}.
 \item[\bf(A3)] The Anderson potential $V_\omega$ is defined by
 \be\label{eq:ranop} V_\omega \psi(x)=\omega_x \psi(x), \ee
 where  $\omega_x$ are i.i.d. random variables, with a probability distribution $\mu$ supported in a compact interval $[a,b]$, with $a,b \in\mathbb R$, $a<b$.  The probability space is constructed as in the previous subsection, $\Omega=[a,b]^{\Gamma}$,
with probability measure $\P=\bigotimes_{x\in\Gamma}\mu$. We denote by $\omega=(\omega_x)_{x\in\Gamma}$ an element of $\Omega$.
\end{enumerate}
For simplicity, we will write $H_\omega$ when $\lambda>0$ is fixed or its value is clear from the context.

Under these assumptions, the operators $-\Delta$ and $V_\omega$ are bounded and self-adjoint, therefore $H_\omega$ is a bounded and self-adjoint Schr\"odinger operator on $\ell^2(\Gamma)$ with spectrum $\sigma(H_\om)$ (for bounded operators this is direct, while in the case of unbounded operators, in the continuous setting $\rm L^2(\mathbb R^d)$, this is a consequence of the Kato-Rellich theorem \cite{RSI}).

In the sequel we will be interested in the cases where $\Gamma=\ZZ^d$, the square lattice, and $\Gamma=\mathbb B$, the Bethe lattice, that is, a tree with constant vertex degree $K$.

Given $z\in\mathbb C\setminus \mathbb R$, we denote the resolvent of $H_\omega$ by $G_\omega(z)=(H_\omega-z)^{-1}$ and write $G_\omega(z;x,y):=\angles{\delta_x,(H_\omega-z)^{-1}\delta_y}$.

\begin{rem}
\begin{enumerate}\quad
\item The model originally proposed by Anderson in \cite{And58} corresponds to \eqref{eq:and} with $\Gamma=\mathbb Z^d$ and probability distribution $\mu$ with probability density $\frac{1}{2}\chi_{[-1,1]}$.
\item The assumption of $\Gamma$ having a bounded vertex degree in $\textrm{\bf(A1)}$ can be relaxed, see \cite{T11} and references therein.
    \item In our discussion we do not include the case of single-site potentials of changing sign. The difficulty in these models is the lack of monotonicity of the eigenvalues of the (finite-volume) operator with respect to the parameter $\omega$ \cite{CaE,ESS,ETV10,ETV11,EKTV}. Another setting we have omitted is quantum graphs, where the methods to prove localization (the Multiscale Analysis and the Fractional Moment Method) can also be applied. For results on localization on diverse types of Anderson models on quantum graphs, see \cite{KloP08,KloP09,ASW06bis,ExHS,Schub,HiP,Sa}.
\end{enumerate}
\end{rem}

\begin{defn}The map $\Omega\ni\omega\mapsto H_\omega\in\mathcal L(\mathcal H)$, the space of linear operators acting on $\mathcal H$, is measurable if for all $\varphi,\psi \in\mathcal H$, the map $\Omega\ni\omega\mapsto \langle \varphi, H_\omega \psi \rangle \in \mathbb C$ is measurable.
\end{defn}
\begin{thm}The Anderson model on $\ell^2(\Gamma)$ is measurable.
 \end{thm}
 \begin{proof} It is enough to show that the map $\Omega\ni\omega\mapsto h_{\varphi}:=\langle \varphi, V_\omega \varphi\rangle \in \mathbb R$ is measurable, for $\varphi\in\ell^2(\Gamma)$, $\varphi\neq 0$. The result then extends to $\langle \varphi, V_\omega \psi\rangle$ for $\varphi,\psi\neq 0$ using the polarization identity.

Let $\mathcal O$ be an open set in $\mathbb R$ such that $h_\varphi^{-1}(\mathcal O)\neq \emptyset$, and take $\omega^0\in h_\varphi^{-1}(\mathcal O)$. We will show that there is a neighborhood of $\omega^0$ whose image under $h_\varphi$ is contained in $\mathcal O$.

Since $h_\varphi(\omega^0)\in\mathcal O$ and $\mathcal O$ is open, there exists $\epsilon>0$ such that for every $z\in\mathbb R$ satisfying $\abs{h_\varphi(\omega^0)-z}<\epsilon$, we have $z\in\mathcal O$. Let $c=\sup_{x\in\Gamma}\abs{\omega_x}<\infty$. Since $\norm{\varphi}<\infty$, there exists a compact set $K=K(\epsilon,\varphi,b)\subset \Gamma$ such that
\be\label{eq:cyl} \sum_{x\in K^c} \abs{\varphi(x)}^2 <\frac{\epsilon}{2c}, \ee
where $K^c=\Gamma\setminus K$. Consider the neighborhood of $\omega^0$ in $\Omega$ defined by
\be B_\epsilon(\omega^0)=\{\omega=(\omega_x)_{x\in\Gamma};\, \abs{\omega_x-\omega_x^0}<\frac{\epsilon}{2\norm{\varphi}} \,\text{for }x\in K,\, \omega_x\in[a,b]\,\text{for }\,x\in K^c  \}, \ee
where $K$ is the compact set defined by \eqref{eq:cyl}. Then, for $\omega\in B_\epsilon(\omega^0)$
\begin{align}
h_\varphi(\omega)& = \sum_{x\in K}\omega_x \abs{\varphi(x)}^2+\sum_{x\in K^c} \omega_x\abs{\varphi(x)}^2\\
& = h_\varphi(\omega^0)+\sum_{x\in K} (\omega_x-\omega^0_x) \abs{\varphi(x)}^2+\sum_{x\in K^c} (\omega_x-\omega^0_x) \abs{\varphi(x)}^2
\end{align}
Thus
\begin{align} \abs{h_\varphi(\omega)-h_\varphi(\omega^0)}&\leq \sup_{x\in K}\abs{\omega_x-\omega^0_x} \sum_{x\in K} \abs{\varphi(x)}^2
+ \sup_{x\in\Gamma} \abs{\omega_x}\sum_{x\in K^c} \abs{\varphi(x)}^2\\
& \leq \epsilon.
\end{align}
This implies that $h_\varphi(B_\epsilon(\omega^0))\subset \mathcal O$, therefore,  $h_\varphi^{-1}(\mathcal O)$ is an open set.
\end{proof}

\subsection{The deterministic spectrum}
We now recall some standard facts from the spectral theory of self-adjoint operators, see e.g. \cite{RSI}, or the notes by S. Gol\'enia in this volume \cite{Go}. Let $H$ be a self-adjoint operator on a Hilbert space $\mathcal H$, with spectrum $\sigma(H)\subset \mathbb R$. For a given $\varphi\in \ell^2(\Gamma)$ there exists a real measure $\mu_{H,\varphi}$ on $\mathcal B(\mathbb R)$ (the Borel sets of $\mathbb R$) such that
\be \angles{\varphi,H,\varphi}= \int_\RR E d\mu_{H,\varphi}(E). \ee
This measure is called \emph{spectral measure} and using Lebesgue's decomposition Theorem, it can be decomposed into three mutually singular parts:
\be \mu_{H,\varphi}=\mu_{H,\varphi}^{pp}+\mu_{H,\varphi}^{sc}+\mu_{H,\varphi}^{ac}, \ee
that are, respectively, pure-point (pp), singular continuous (sc), and absolutely continuous (ac) with respect to the Lebesgue measure. This induces a decomposition of the Hilbert space $\mathcal H=\mathcal H_{pp}\oplus\mathcal H_{sc}\oplus\mathcal H_{ac}$, where
\be\mathcal H_{*}=\{\varphi \in\mathcal H;\, \mu_{H,\varphi}= \mu_{H,\varphi}^{*}\},\quad \text{for}\,*\in \{pp,sc,ac\}.\ee
By restricting the operator $H$ to each of these spaces, we obtain the following decomposition of the spectrum
\be \sigma(H)=\sigma_{pp}(H)\cup \sigma_{sc}(H)\cup \sigma_{ac}(H). \ee

Note that the random operator $H_\omega$ represents a \emph{family} of operators $(H_\omega)_{\omega\in\Omega}$ acting on $\mathcal H=\ell^2(\Gamma)$, and each realization of the operator has a spectral decomposition as above. We will see in what follows that the Anderson model satisfies a fundamental property, called \emph{ergodicity}, by which the spectrum and its spectral parts are deterministic, that is, do not depend on the realization $\omega\in\Omega$. Therefore, we aim for spectral results that hold for almost every $\omega\in\Omega$.

\begin{defn}\label{def:erg}
\begin{itemize}
\item[i)]Given a probability space $(\Omega,\mathcal B,\P)$, a measurable application $\tau:\Omega\rightarrow \Omega$ is called measure preserving if $\P(\tau^{-1}B)=\P(B)$ for all $B\in\mathcal B$.
    \item[ii)]Given a collection $(\tau_\gamma)_{\gamma\in\Gamma}$ of measure preserving transformations, we call $B\in\mathcal B$ invariant with respect to it if $\tau_\gamma^{-1}B=B$ for all $\gamma\in\Gamma$.
        \item[iii)] A measure preserving group of transformations $(\tau_\gamma)_{\gamma\in\Gamma}$ is called \emph{ergodic} with respect to $\P$ if any set $B\in\mathcal B$ that is invariant with respect to the family $(\tau_\gamma)_{\gamma\in\Gamma}$ has probability either zero or one.
\end{itemize}
\end{defn}

\begin{defn}The operator $H_\om$ is called \emph{ergodic} if there exists an ergodic group of transformations  $(\tau_\gamma)_{\gamma\in \Gamma}$ acting on $\Omega$ associated to a family of unitary operators  $(U_\gamma)_{\gamma\in \Gamma}$ on $\mathcal H$ such that
\be\label{eq:erg}
 H_{\tau_\gamma (\omega)}= U_\gamma H_{\omega}U_\gamma^* \quad\text{for all}\,\gamma\in \Gamma,
\ee
where $U_\gamma^*$ denotes the adjoint of $U_\gamma$.
\end{defn}
The notion of ergodicity applies to more general Schr\"odinger operators, see \cite{Dam16} and \cite{PF92}.
%Property \eqref{eq:erg} is called \emph{covariance}.
\smallskip

\noindent{\bf Examples}
\begin{itemize}
\item[a)] The Anderson model $H_\omega$ on $\ell^2(\mathbb Z^d)$ is ergodic with respect to $\mathbb Z^d$.
 That is, with respect to the translations $\tau_\gamma(\omega)=(\omega_{x+\gamma})_{x\in\mathbb Z^d}$ and $U_\gamma\varphi(x)=\varphi(x-\gamma)$ with $\gamma \in \mathbb Z^d$. Note that $U_{\gamma}^*$ is given by $U_{\gamma}^*\varphi(x)=\varphi(x+\gamma)=U_{-\gamma}$.
Then
\begin{align}
U_\gamma H_\omega U_{-\gamma}\varphi (x)& =U_\gamma\, (-\Delta) U_{-\gamma}\,\varphi (x)+ U_\gamma\,\left( V_\omega\, U_{-\gamma}\right)\,\varphi(x)\notag\\
& = -\Delta\,\varphi(x)+\left(V_\omega \,U_{-\gamma}\,\varphi\right)(x-\gamma)\notag\\
 & = -\Delta\,\varphi(x)+ V_\omega(x-\gamma)\,\left(U_{-\gamma}\varphi\right)(x-\gamma)\notag\\
 & = -\Delta\,\varphi(x)+ V_\omega(x-\gamma)\varphi(x).\notag
\end{align}
Since $V_\omega(x-\gamma)\varphi(x)=\omega_{x-\gamma}\varphi(x)=V_{\tau_\gamma(\omega)}\varphi(x)$, we have
\be  U_\gamma H_\omega U_{-\gamma}\varphi= H_{\tau_\gamma(\omega)}. \ee

\item[b)] Consider the Bethe lattice $\mathbb B$ with vertex degree $K\in\mathbb N$ even, rooted at zero (also called the Cayley tree).  The Anderson model $H_\omega$ on $\ell^2(\mathbb B)$ is ergodic with respect to the translations generated by the free group generating $\mathbb B$. For example, consider the case $K=4$, where $\mathbb B$ is generated by the family $\mathcal T=\{a,b,a^{-1},b^{-1}\}$. Denote by $0$ the root, such that $0=aa^{-1}=a^{-1}a=bb^{-1}=b^{-1}b$. The set $\mathcal T \cup \{0\}$ forms a group with respect to multiplication to the right, defined by $\alpha_x(y)=yx$. If we define the operator $U_x\varphi(y)=\varphi\circ \alpha_x(y)=\varphi(yx)$ for $x\in \mathcal T$, $y\in\mathbb B$, we have that $U^{-1}_x=U^*_x=U_{x^{-1}}$. Then any point $y\in\mathbb B$ can be written as a unique composition of elements in $\mathcal T$, that is, if $d_{\mathbb B}(0,x)=n$, $x$ can be written uniquely as a product $x=\Pi_{i=1}^n e_i$, where $e_i\in\mathcal T$, see Fig. \ref{bethe}. We define in an analogous way $\tau_x(\omega)=(\omega_{xy})_{y\in\mathbb B}$ for $x\in \mathcal T$. In this way we have $U_x H_\om U_{x^{-1}}=H_{\tau_x(\om)}$.
     See also \cite[Appendix]{AK92} for another approach to these translations.\\
\begin{figure}[h]
  \centering
  \includegraphics[width=11cm]{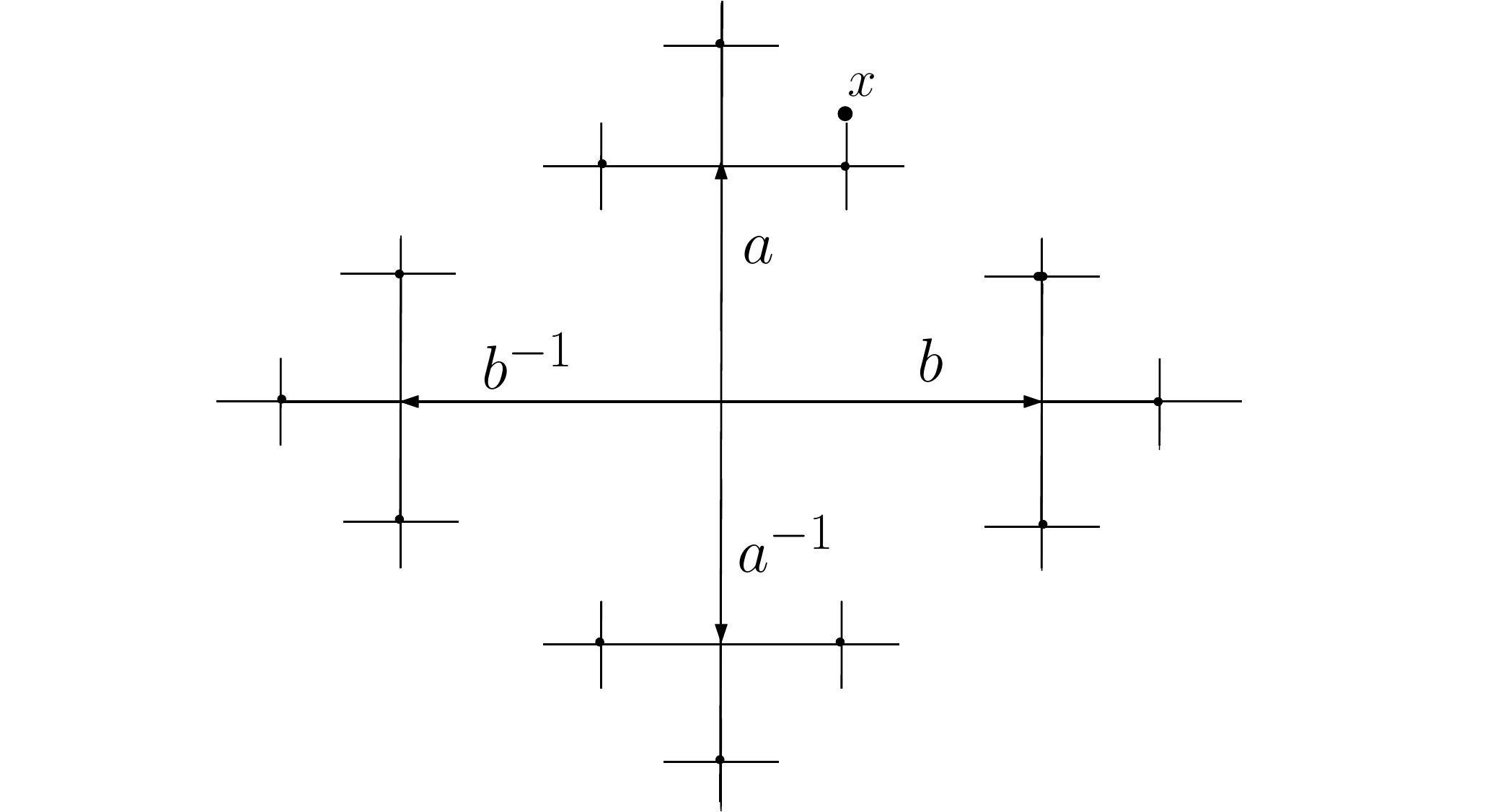}\\
  \caption{In this picture, $x=aba$. The decreasing length of the edges in the graph is only for pedagogical purposes. It simulates an optical effect to depict growing distance from the root.}
  \label{bethe}
\end{figure}

\end{itemize}

If a random variable $X$ is invariant under an ergodic family of transformations $(\tau_\gamma)_{\gamma\in\Gamma}$, that is, if $X(\tau_\gamma\om)=X(\om)$ for all $\gamma\in\Gamma$, then $X$ is almost surely constant, see e.g. \cite[Proposition 9.1]{CyFKS}. Since the spectral projection $\chi_I(H_\omega)$ associated to an interval $I\subset \mathbb R$  is a measurable function of $H_\omega$, the ergodicity of $H_\omega$ implies that $\chi_I(H_\omega)$, seen as a random variable, is almost surely constant on $\Omega$. On the other hand, this function characterizes the spectrum of the operator \cite[Theorem 7.22]{Weid}:
\be \sigma(H_\omega)=\{ E\in\mathbb R; \chi_{(E-\epsilon,E+\epsilon)}(H_\omega)\neq 0 \, \mbox{for every}\, \epsilon>0 \}. \ee
Therefore, the almost-sure constancy of the spectral projection implies that the spectrum is deterministic. The same applies to the spectral projections $\chi_I^{(*)}(H_\omega)$ associated to the different spectral types $*\in\{pp,sc,ac\}$:
\begin{thm}[\cite{Pas80,KuSo80,KMa}]Let $\Omega$ be a probability space and $H_\omega$ with $\omega\in\Omega$ an ergodic operator. There exist closed sets $\Sigma$, $\Sigma_{pp}, \Sigma_{ac},\Sigma_{sc}\subset \RR$ such that for  $\P$-a.e. $\omega\in\Omega$,
\[ \Sigma=\sigma(H_\omega),\]
\[\Sigma_{pp}=\sigma_{pp}(H_\omega),\,\,\, \Sigma_{ac}=\sigma_{ac}(H_\omega),\,\,\,  \Sigma_{sc}=\sigma_{sc}(H_\omega). \]
where \emph{pp, ac} and $sc$ stand for the pure point, absolutely continuous and singular continuous part of the spectrum.
\end{thm}

Knowing that the spectrum of the operator is deterministic is not enough for our purposes, since we will later work in particular regions of the spectrum. We also need to know its location. In the case of the Anderson model on the lattice, Kunz and Souillard \cite{KuSo80} proved the following as a consequence of ergodicity,

\begin{thm}\label{t:ks}Let $H_\omega=-\Delta+\lambda V_\omega$ with $\lambda>0$ be the Anderson model acting on $\ell^2(\ZZ^d)$. Then
\be\label{eq:ks}\sigma(H_\om)=\sigma(-\Delta)+\lambda\supp \mu ,\quad \text{for a.e.}\, \omega\in\Omega\ee
where the sum of sets is defined as $A+B=\{a+b\in\mathbb R; a\in A,\,b\in B\}$. \end{thm}

The proof of this theorem can be found in the standard references, e.g. \cite{KuSo80,K,S10}, but because it is rather simple and illustrates the idea behind ergodicity, we include it here. First, we need some auxiliary results. The following result corresponds to \cite[Proposition 3.8]{K},
\begin{prop}\label{prop:what}There exists a set $\Omega_0\subset\Omega$ with $\P(\Omega_0)=1$ such that given any $\omega\in\Omega_0$, a finite set $\Lambda\subset\ZZ^d$, any sequence $v=(v_i)_{i\in\ZZ^d}$ with $v_i\in\supp\mu$, and any $\epsilon >0$, there exists a sequence of vectors $(x_n)_{n\in\ZZ^d}\subset\ZZ^d$ with $\norm{x_n}_\infty\rightarrow\infty$ such that
\be \sup_{i\in\L}\abs{v_i -V_\omega(i+x_n)}<\epsilon. \ee
\end{prop}
\begin{proof}
Let $\Lambda\subset\mathbb Z^d$ be a finite set,
{$(v_i)_{i\in\ZZ^d}$ with $v_i\in\rm{supp}\,\mu$} and $\epsilon>0$. Take a sequence $(x_n)_{n\in\ZZ^d}$ such that $\norm{x_n-x_m}_\infty>\rm{diam}\, (\Lambda)$ for $n\neq m$. We write $\L(x_n):=\L+x_n$, and consider the events
\be \label{eq:A1}A_{\eps,v}^\L(x_n):=\{\omega\in\Omega:\,\sup_{i\in\Lambda(x_n)}\abs{V_\omega(i+x_n)-v_i}<\epsilon\}.\ee
Since the $\omega_i$ are i.i.d., the events $A_{\eps,v}^\L(x_n)$ and $A_{\eps,v}^\L(x_m)$ are independent for $n\neq m$ and moreover, since $v_i\in \rm{supp}\,\mu$, we have that $\mathbb P(A_{\eps,v}^\L(x_n))=\mathbb P(A_{\eps,v}^\L(0))>0$ for all $n$, therefore
\be \sum_n \mathbb P(A_{\eps,v}^\L(x_n))=\infty.\ee
From Lemma \ref{t:bc} (Borel-Cantelli) we deduce that
\be\label{eq:A2}\mathbb P\left(A_\infty\left(\Lambda,v,\epsilon\right)\right)=1,\ee
where
\begin{align} A_\infty\left(\Lambda,v,\epsilon\right)& := \bigcap_N\bigcup_{n\geq N} A_{\eps,v}^\L(x_n)\\
&= \{\omega\in\Omega:\,\omega\in A_{\eps,v}^\L(x_n)\text{ for infinitely many }n\}.
\end{align}
Now, we want to take $\L$ in the space $F$ of all possible finite subsets of $\mathbb Z^d$, which is countable.
%Then
%\[ \mathbb P\left( \bigcap_{\Lambda\in F}A_\infty(\Lambda,v,\epsilon) \right)=1.  \]
 In the same way, we want to consider all possible sequences $(v_i)$ with $v_i\in \rm{supp }\, \mu$ and all $\epsilon>0$. Taking a dense countable subset $Q\subset\rm{supp }\, \mu$ and $\epsilon=1/k$ with $k\in\NN$, we get
\be \mathbb P(\Omega_0):=\mathbb P\left( \bigcap_{\substack{v\in Q^{\ZZ^d}\\ \Lambda\in F,\,k\in\NN}}A_\infty(\Lambda,v,1/k)\right)=1,\ee
where we used the fact that a countable intersection of sets of probability one is of probability one.
\end{proof}

We call a configuration $\omega$ restricted to a compact set a \emph{pattern}. Proposition \ref{prop:what} says that for any given pattern, we can find a sparse sequence of sets in space of the same size such that in those sets the potential looks \emph{almost} like the original pattern, see Fig. \ref{patt}. The proof can be found in \cite[Proposition 3.8]{K} and relies on one hand on the underlying lattice structure of the potential, and the fact that the random variables on each site $\omega_i$ are independent and identically distributed. This implies that events defined by the potential looking like a prescribed configuration in disjoint regions of space are independent and have the same probability. W. Kirsch describes this in \cite{K} as "whatever can happen will happen, in fact infinitely often".

\begin{figure}[h]
  \centering
  \includegraphics[width=10cm]{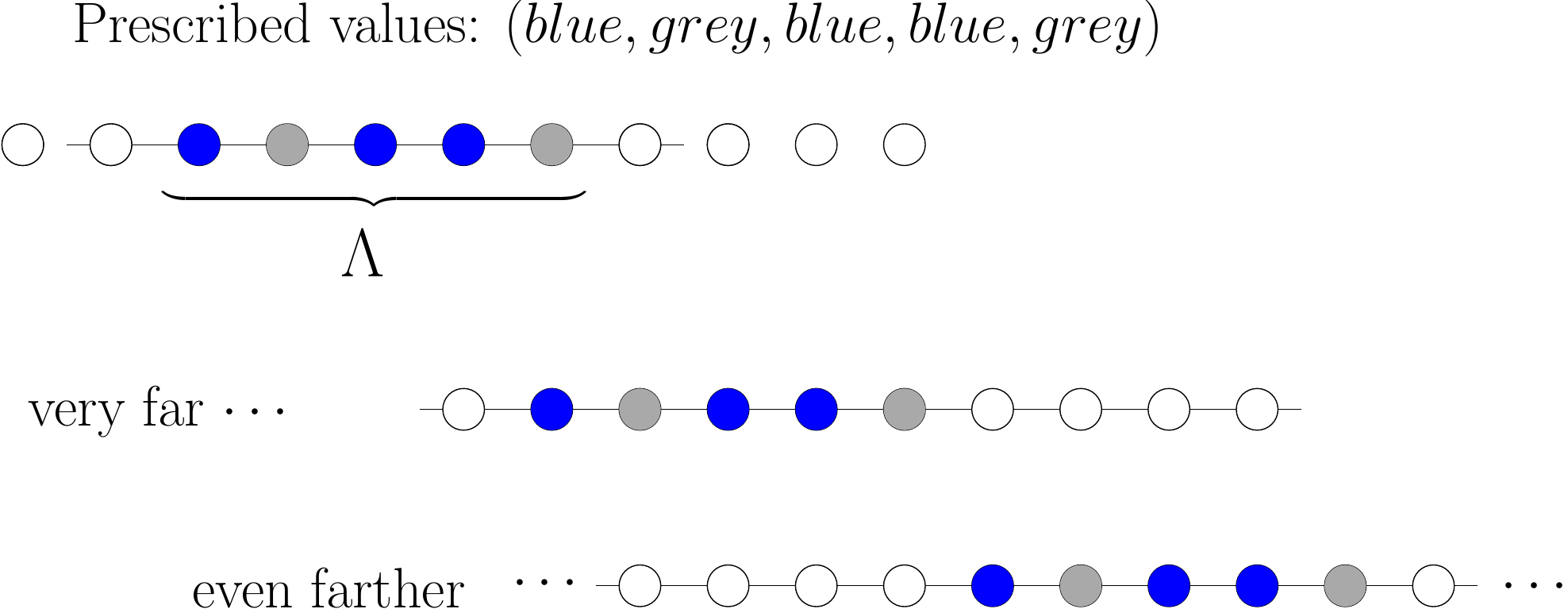}\\
  \caption{Example in $\mathbb Z$ with $\omega_j=\{blue,grey\}$.}
  \label{patt}
\end{figure}

\smallskip

The following result, known as  Weyl's Criterion, is very useful when we need to determine the spectrum of a self-adjoint operator \cite[Theorem 7.22]{Weid}, see also \cite[Theorem 7.6]{Go}.

\begin{thm}\label{t:weyl} Let $H$ be a self-adjoint operator on a Hilbert space $\mathcal H$, with core $\mathcal H_0$, then $E\in\sigma(H)$ if and only if there exists a sequence $(\varphi_n)_{n\in\NN}\subset \mathcal H_0$ such that $\liminf_n\norm{\varphi_n}>0$ and $\norm{(H-E)\varphi_n}\rightarrow 0$ when $n$ tends to infinity.
\end{thm}

\begin{rem} The sequence $(\varphi_n)_{n\in\NN}$ is called a Weyl sequence for $H$ and $E$.
\end{rem}
\begin{proof}[Proof of Theorem \ref{t:ks}]
Since $V_\omega$ is a multiplication operator, its spectrum is given by the closure of its essential range, $\sigma(V_\omega)=\overline{\{\omega_x;x\in\ZZ^d\}}=\supp\mu$ almost-surely. Using \cite[Theorem V.4.10]{Ka}
we get
\be \sigma(-\Delta+\lambda V_\omega)\subseteq \sigma(-\Delta)+\lambda\supp\mu. \ee
In order to prove the converse, let $E=E_1+E_2$ with $E_1\in \sigma(-\Delta)$ and $E_2\in \lambda\supp\mu$. Since $\ell_c^2(\ZZ^d)$ is a core for $-\Delta$, by Theorem \ref{t:weyl}, there exists a Weyl sequence $\varphi_n\in\ell^2_c(\ZZ^d)$ for $-\Delta$ and $E_1$ such that $\norm{(-\Delta+E_1)\varphi_n}\rightarrow 0$. Define $\psi_{n,\gamma}=\varphi_n(\cdot-\gamma)$ for $\gamma\in\ZZ^d$ and note that $-\Delta$ is invariant under translations by elements $\gamma\in\ZZ^d$. Therefore $\psi_{n,\gamma}$ is also a Weyl sequence for $-\Delta$ and $E_1$.
Now, let $\Lambda_n\subset \ZZ^d$ be a box containing $\supp\varphi_n$. Since $E_2\in\lambda\supp\mu$, Proposition \ref{prop:what} gives the existence of a set $\Omega_0$ of probability one for which for every $\omega\in\Omega_0$, there exists a sequence $(\gamma^{(n)}_m)_{m\in\ZZ^d}$ with $\norm{\gamma^{(n)}_m}_\infty\rightarrow \infty$ when $m\rightarrow \infty$ such that
\be\label{eq:what} \sup_{i\in\Lambda_n}\abs{E_2-\lambda V_\om(i+\gamma_m^{(n)})}<\frac{1}{n} \ee
We define $\gamma_n:=\gamma_n^{(n)}$, and consider the sequence $\psi_n=\psi_{n,\gamma_n}$. Then
\begin{align} \norm{(-\Delta+\lambda V_\om-E)\psi_n}& \leq\norm{(-\Delta-E_1)\psi_n}+\norm{(\lambda V_\om-E_2)\psi_n}.
\end{align}
The first term in the r.h.s tends to zero when $n$ tends to infinity because $\psi_n$ is a Weyl sequence for $-\Delta$, while the second term tends to zero due to \eqref{eq:what}. Therefore $\psi_n$ is a Weyl sequence for $H_\omega$ and $E$ for $\om\in\Omega_0$, so by Theorem \ref{t:weyl}, $E\in \sigma(H_\omega)$ almost surely.
\end{proof}

Note that in the proof of Theorem \ref{t:ks} the Weyl sequence $\psi_n$  can be chosen to be orthogonal and normalized, therefore \cite[Theorem 7.24]{Weid} implies
\begin{prop} The operator $H_\omega$ has purely essential spectrum.
\end{prop}

We see from the previous result that in order to determine the almost-sure spectrum $\Sigma$ of $H_\omega$, we need to determine first $\sigma(-\Delta)$. In the case $\Gamma$ is of constant degree, the difference between $\sigma(-\Delta)$ and  $\sigma(\mathcal A_{\Gamma})$ is a constant $\deg_\Gamma$. Since shifting the spectrum by a constant does not change the spectral properties, one can absorb it in the potential and therefore the problem turns into determining the spectrum of the adjacency matrix $A_\Gamma$. In a (common) abuse of notation, we consider from now on the operator $H_\omega=\mathcal A_{\Gamma}+V_\omega$. From \cite[Sections 4.2 and 4.3]{Go} we have that
\begin{itemize}
\item[a)]For $\Gamma=\mathbb Z^d$, $\sigma(\mathcal A_{\ZZ^d})=[-2d,2d]$. Therefore, for the Anderson model $H_\omega$ on $\ell^2(\mathbb Z^d)$, we have \be\sigma(H_{\om})=[-2d,2d]+\lambda\,\textrm{supp}\mu.\ee
\item[b)]For $\Gamma=\mathbb B$, $\sigma(\mathcal A_{\mathbb B})=[-2\sqrt{K},2\sqrt{K}]$, where $K+1$ is the degree of $\mathbb B$. Therefore, for the Anderson model $H_{\omega}$ on $\ell^2(\mathbb B)$, \be \sigma(H_{\om})=[-2\sqrt{K},2\sqrt{K}]+\lambda\,\textrm{supp}\mu\ee.
\end{itemize}

\subsubsection{Non-ergodic operators}\label{s:ne}
The proof of Theorem \ref{t:ks} using Weyl sequences relies on the covariance property \eqref{eq:erg}. However,
in the case where the operator is not ergodic in the sense of \eqref{eq:erg}, under certain conditions a weaker result holds, which can be enough for our purposes. Such is the case of an Anderson operator in which the potential is forced to be zero in some prescribed regions of $\ZZ^d$ in such a way that it looses the underlying periodic structure of the lattice (see \cite{RM13,EK}). The Anderson model with missing sites enters in the more general setting of the so-called Delone-Anderson operators.

\begin{figure}[h]
  \centering
  \includegraphics[width=8cm]{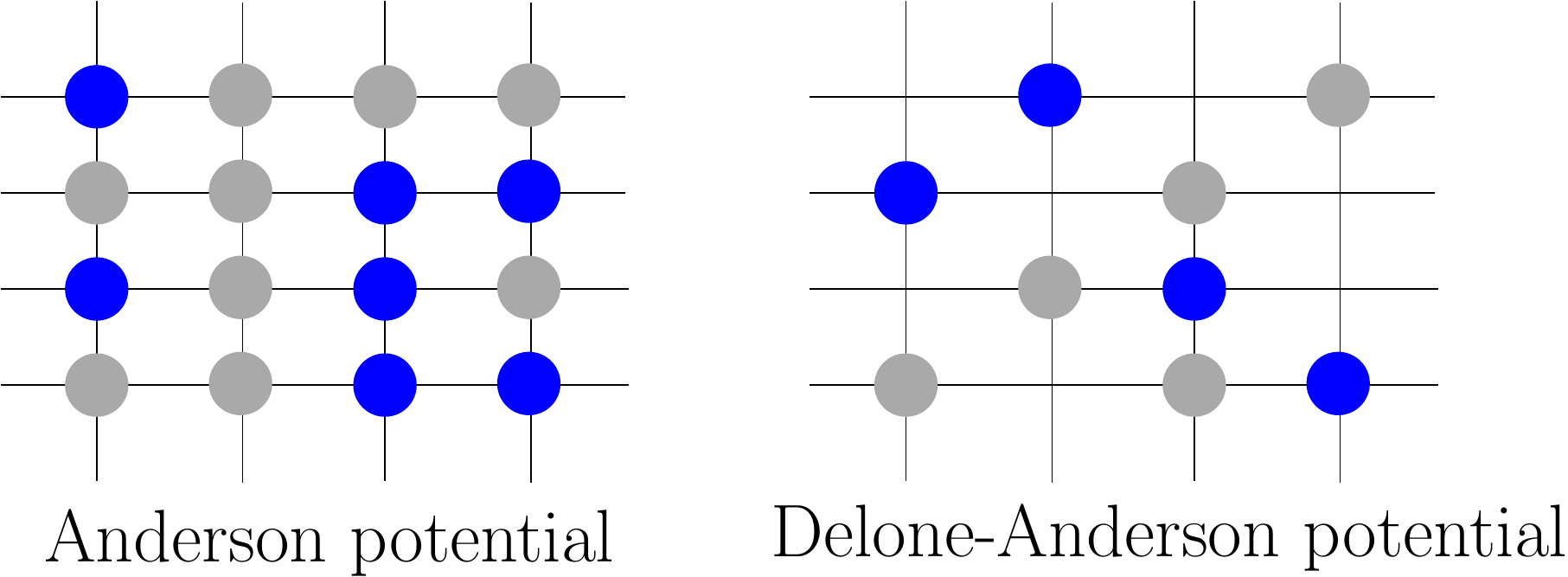}\\
  %\caption{}
  %\label{}
\end{figure}

For this model we have the following
\begin{thm}\label{t:nerg} Let $D$ be a discrete subset of $\ZZ^d$ and $H_\dom=-\Delta+V_\dom$ acting on $\ell^2(\ZZ^d)$, where
\be V_\dom(x)=\begin{cases}\omega_x & x\in D \\ 0  & x\notin D \end{cases}.
\ee
Assume the following:
\begin{itemize}
\item[i.] The set $D$ is a Delone set, that is, there is a constant $R$ such that $[-R+x,R+x]^d\cap D\neq \emptyset$ for all $x\in\ZZ^d$.
 \item[ii.] $\omega_x$ are i.i.d. random variables with common probability distribution $\mu$, compactly supported with $0\in\supp\mu$.
\end{itemize}
Then,
\be \sigma(-\Delta)\subseteq \sigma(H_\omega) \,\text{ for a.e. }\omega\in\Omega.\ee
\end{thm}

\begin{proof}
We follow the proof of Theorem \ref{t:ks} with $E\in\sigma(-\Delta)$, using Proposition \ref{prop:what} with the particular choice $v_i=0$ $\forall i\in\mathbb Z^d$, and taking a finite set $\Lambda$ such that $[-R,R]^d\cap \ZZ^d\subset \Lambda$ see \cite[Section 6.4]{RM12}. In this case, the events \eqref{eq:A1} are, for the Delone set $D$,
\be \tilde A_{\epsilon,0}^\L(x_n)=\{\omega\in\Omega;\,\abs{\om_i}< \epsilon,\,\forall i\in\Lambda(x_n)\cap D\},\ee
where $\Lambda(x_n)=\Lambda+x_n$ as before. Then the events $\tilde A_{\epsilon,0}^\L(x_n)$ are independent, and $\P(\tilde A_{\epsilon,0}^\L(x_n))>0$ for any $n$. These events do not have the same probability, since we do not necessarily have $\sharp(\Lambda(x_n)\cap D) = \sharp(\Lambda(x_m)\cap D)$ for $n\neq m$. However, the following holds for $\Lambda$ uniformly in the position $x_n$
\be 1\leq C_{R,d,\Lambda}\leq \sharp (\L(x_n)\cap D)\leq C_{d,\Lambda} \ee
for positive constants $C_{R,d,\Lambda}$ and $C_{d,\Lambda}$ that depend on $R$ and the volume of the set $\L$. Therefore, for all $n$
\begin{align}
\P(A_{\epsilon,0}^\L(x_n) )&= \P(\abs{\omega_j}<\epsilon)^{\sharp (\L\cap D)}\\
&\geq \mu((-\epsilon,\epsilon))^{C_{d,\L}},
\end{align}
which implies \eqref{eq:A2} and the rest of the proof follows. We see that Proposition \ref{prop:what} holds even though the potential is not ergodic, because although  the probability of the event that the potential has values near zero on a compact set $\L$ might depend on the position of set $\L$, there is a lower bound that is uniform in the position, which is enough for the proof.
\end{proof}
In the case $\supp \mu\subseteq[0,\infty)$, Theorem \ref{t:nerg} implies we recover a portion of almost-sure spectrum at the lower spectral edge, which is enough to give a meaningful notion of the phenomenon of localization that will be described in Section \ref{s:loc}. For localization results on the Anderson model with missing sites, see \cite{RM13,EK,ES}.

\begin{rem}
The notion of ergodicity and the almost-sure constancy of the spectrum and spectral types holds in more general settings. For example, when the set $\Gamma$ changes with $\om$ and the operator $H_\om$ acts on the Hilbert space $\ell^2(\Gamma_\om)$. This problem is studied in an algebraic setting in \cite{LPV07}. For the particular case where $\Gamma_\om$ is a random graph subset of $\ZZ^d$ obtained through a bond percolation process, see \cite{Ve05a,Ve05b,KMue,MueS}. The case of operators on manifolds with random metrics and random potentials is treated in \cite{LPV03}. The study of the deterministic nature of the spectrum is often related to the study of ergodic theorems and the existence of the integrated density of states. In this field, a considerable part of the literature has been produced by the Chemnitz School (D. Lenz, P. Stollmann, I. Veseli\'c, M. Keller, M. Tautenhahn, F. Schwarzenberger and C. Schumacher, among others).
\end{rem}

\section{Types of localization}\label{s:loc}
The efforts to explain rigorously the phenomenon of Anderson localization led to the development of the theory of random
Schr\"odinger operators, of which the Anderson model $H_\omega$ is one example.
%:
%\be H_\omega=-\Delta+\sum_{n\in\ZZ^d}\omega_n\angles{\delta_n,\cdot}\delta_n.\ee
The first mathematical results on localization for random operators showed the existence of pure point spectrum in the one-dimensional setting \cite{GMP77,Pas80}. Stronger results showing in addition to pure point spectrum the exponential decay of eigenfunctions (Anderson localization) were obtained in \cite{KuSo80}. The first results for dimension $d>1$ were given by Fr\"ohlich and Spencer \cite{FrS83} showing the absence of diffusion uniformly in $t$:
\be\label{eq:ad} \sum_x\norm{x}^2\abs{\angles{\delta_x,e^{-itH_\omega}\delta_0}}^2<\infty,\quad\text{for a.e. }\om\in\Omega. \ee
Their work set the foundations of the Multiscale Analysis (MSA), a method to prove localization by an induction over finite-volume versions of the resolvent of the operator. Martinelli and Scoppolla \cite{MS} improved on the technique of \cite{FrS83} and were able to prove the absence of absolutely continuous spectrum, while the joint efforts of the aforementioned authors gave way to the first proof of Anderson localization in arbitrary dimension in \cite{FrMSS}. A stronger notion of localization called dynamical localization (see below) followed from the work of Aizenman and Molchanov \cite{AM93} and Aizenman \cite{A}. In these works, the authors developed the Fractional Moment Method (FMM), that together with the MSA are the two available techniques to prove localization in dimension $d\geq 1$. It was not until the 90s with the work of del Rio, Jitomirskaya, Last, and Simon \cite{dRJLS95} that it was understood that the definitions of localization were not equivalent and that what was called Anderson localization was not enough to describe the absence of quantum transport observed by Anderson in disordered materials.

The aforementioned results involve probability measures that are regular, i.e., at least of H\"older continuity. For the Anderson model with Bernoulli random variables, although it is expected that the model exhibits localization, the lack of regularity of the probability measure makes it a technically difficult problem. For the Bernoulli Anderson model in one dimension, acting on $\ell^2(\ZZ)$, the first proof of localization was obtained in \cite{CKM}, while for a continuous analogous of the Anderson model acting on $\Lp{\RR}$ a first proof was given in \cite{DSS}. In the continuous case Anderson localization was shown for dimensions $d\geq1$ in the ground-breaking work \cite{BoK05}. The authors developed a multi-scale scheme and relied on unique continuation properties of the eigenfunctions of finite-volume operators. In \cite{GK11} this MSA was improved to obtain dynamical localization.

At a more advanced level, many intertwined notions of localization can be defined, see \cite[Section 3]{Kl} for a zoology of localization types. For the reader's convenience, we summarize the most relevant ones studied in the literature.
\begin{defn}\label{def:loc} Let $H_\omega$ be a random Schr\"odinger operator acting on $\ell^2(\Gamma)$.
\begin{itemize}
\item[i.] We say that the operator $H_\omega$ exhibits \emph{spectral localization} in an interval $I$ if $\sigma(H_\om)\cap I=\sigma_{pp}(H_\om)\cap I$, almost surely.
    \item[ii.] We say that the operator $H_\om$ exhibits \emph{Anderson localization} (AL) in $ I$ if
$\sigma(H_\om)\cap  I=\sigma_{pp}(H_\om)\cap I$ with exponentially decaying eigenfunctions, almost surely.
\item[iii.] We say that $H_\omega$ exhibits \emph{dynamical localization} (DL) in $ I$ if there exist constants $0<c,C<\infty$, and $\zeta\in(0,1]$ such that for all $x,y\in\Gamma$,
\be\label{eq:dl}\mathbb E\left( \sup_{t\in\mathbb R}\abs{\langle \delta_y, e^{-itH_\omega}\chi_{ I}(H_\omega)\delta_x\rangle} \right)\leq Ce^{-c\textrm{dist}_\Gamma(x,y)^\zeta}. \ee
\end{itemize}
\end{defn}
We recall \cite[Theorem 8.5]{K}, whose proof is based on the RAGE Theorem (see \cite[Section 5.4]{CyFKS} or \cite{Go} in this volume).
\begin{thm}\label{t:pp}Dynamical localization implies pure point spectrum.
\end{thm}
\begin{defn} We say that $H_\omega$ exhibits \emph{absence of transport} in an interval $I\subset \RR$ if, for any $p\geq 0$ and any $\varphi\in\mathcal H$ with compact support, the following holds almost surely,
\be\label{eq:at}\sup_{t\in\mathbb R}\norm{\abs{X}^p e^{-itH_\omega}\chi_{I}(H_\omega)\varphi}<\infty,  \ee
where $\abs{X}$ is the multiplication operator defined by $\abs{X}\varphi(x)=\norm{x}\varphi(x)$ and $\chi_I(H_\omega)$ is the spectral projection of $H_\om$ associated to the interval $I$.
\end{defn}
Absence of transport means that, as time evolves, wave packets do not propagate in the medium. Because the definition of absence of transport is very strong (it holds for any $p\geq0$), one can show that it also implies pure point spectrum. This is in fact the proof given in \cite[Theorem 8.5]{K}. We retrieve the previous theorem from the following result, whose proof we take from \cite[Section 3]{S10},
\begin{thm}Dynamical localization implies absence of transport.
\end{thm}
\begin{proof}
Take $\varphi\in\ell^2_c(\mathbb Z^d)$, that is, for some $R>0$, $\varphi(x)=0$ for $\norm{x}>R$. Then, since $(\delta_n)$ is an orthonormal base of $\ell^2(\ZZ^d)$, we have
%using the expression
%\[ \norm{x}=\sum_n\abs{\langle x,\delta_n \rangle}^2\]
\begin{align}
\norm{\abs{X}^p e^{-itH_\omega}\chi_I(H_\omega)\varphi}^2 & =\sum_{j\in\mathbb Z^d}\abs{\langle \delta_j,\abs{X}^p e^{-itH_\omega}\chi_I(H_\omega) \varphi \rangle  }^2 \notag\\
&\leq \sum_j \abs{j}^{2p}\abs{\langle \delta_j, e^{-itH_\omega}\chi_I(H_\omega) \varphi \rangle  }^2 \notag\\
%&\leq \sum_j \abs{j}^{2p}\abs{\langle \delta_j, e^{-itH_\omega}\chi_I(H_\omega) \varphi \rangle  }\norm{\varphi} \notag\\
&\leq \sum_j \abs{j}^{2p}\norm{\varphi} \abs{\langle \delta_j, e^{-itH_\omega}\chi_I(H_\omega) \left(\sum_{\abs{k}\leq R} \langle \varphi,\delta_k \rangle \delta_k\right) \rangle  } \notag\\
&\leq \sum_j \sum_{\abs{k}\leq R} \abs{j}^{2p}\norm{\varphi}^2\abs{\langle \delta_j, e^{-itH_\omega}\chi_I(H_\omega)\delta_k  \rangle  }. \notag\\
\notag
\end{align}
We see that \eqref{eq:dl} implies that the r.h.s is summable and uniformly bounded with respect to time, which gives the desired the result.
\end{proof}

\begin{rem}\label{r:loc} \begin{itemize}
\item[i.] Dynamical localization and absence of transport in the sense of \eqref{eq:at} are actually equivalent, as shown in  \cite[Theorem 4.2]{GKduke}.
\item[ii.] Note that absence of diffusion \eqref{eq:ad} corresponds to \eqref{eq:at} with $p=1$.
\end{itemize}
\end{rem}
Dynamical localization also implies that eigenfunctions associated to the pure point spectrum are exponentially decaying. The decay of eigenfunctions observed in the Anderson model is described by a strong notion called SULE (semi-uniform localization of eigenfunctions, see below). Here, we recall a weaker result from \cite[Theorem 9.22]{CyFKS} with a conceptually simple proof to see how dynamical localization implies decay of eigenfunctions. We reproduce its proof here with only slight modifications.
\begin{thm}\label{t:ef}Dynamical localization implies the (sub)exponential decay of eigenfunctions. Namely, suppose that for a given interval $\cI$, there exist positive constants $c,C$ and $\zeta\in(0,1]$ such that \eqref{eq:dl} holds for all $x,y\in\Gamma$ . Then,  for $\P$-a.e. $\omega\in\Omega$, for any $\epsilon>0$, there exists a positive constant $C_{\om,\epsilon}$ such that any eigenfunction of $H_\omega$ in $\cI$ satisfies
\be \abs{\varphi_\omega(x)}\leq C_{\omega,\epsilon}e^{-(c-\epsilon)\dist_\Gamma(x,x_0)^\zeta}, \ee
where $x_0$ is a center of localization of $\varphi_\omega$, i.e., $\abs{\varphi_\omega(x_0)}=\sup_{x\in\Gamma}\abs{\varphi_\omega(x)}$.
\end{thm}
\begin{proof} For simplicity, we drop the subscript $\Gamma$ from the notation $\dist_\Gamma$.
Define the event
\be B_{x,y}:=\left\{\omega\in\Omega;\,\sup_{t\in\mathbb R}\abs{\langle \delta_y, e^{-itH_\omega}\chi_{\mathcal I}(H_\omega)\delta_x\rangle}>e^{-(c-\epsilon)\textrm{dist}(x,y)^\zeta} \right\}. \ee
By Chebyshev's inequality, using \eqref{eq:dl}, we have
\begin{align}
\P\left( B_{x,y}\right) & \leq e^{(c-\epsilon)\textrm{dist}(x,y)^\zeta}\mathbb E\left(\sup_{t\in\mathbb R}\abs{\langle \delta_y, e^{-itH_\omega}\chi_{\mathcal I}(H_\omega)\delta_x\rangle} \right) \\ \notag
& \leq Ce^{-\epsilon\textrm{dist}(x,y)^\zeta}.
\end{align}
Then, for fixed $x$ we have $\sum_y \P(B_{x,y})<\infty$, therefore we can use Theorem \ref{t:bc} (Borel-Cantelli) and obtain
\be \P\left( \bigcup_{N\in\NN}\bigcap_{\norm{y-x}=N}^\infty B_{x,y}^c \right)=1. \ee
This implies that for $\P$-a.e. $\om\in\Omega$, there exists $N_0$ such that $\omega\in B_{x,y}^c$ for all $y$ such that $\textrm{dist}(x,y)\geq N_0$. In the case $\textrm{dist}(x,y)<N_0$, note that
\begin{align} \sup_{t\in\mathbb R}\abs{\langle \delta_y, e^{-itH_\omega}\chi_{\mathcal I}(H_\omega)\delta_x\rangle} & \leq 1\\
& \leq C_{c,N_0,\zeta,\epsilon,\omega}e^{-(c-\epsilon)\textrm{dist}(x,y)^\zeta},
\end{align}
where $C_{c,N_0,\zeta,\epsilon,\omega}$ denotes a constant depending on the parameters ${c,N_0,\zeta,\epsilon,\omega}$.
Therefore, for all $y\in\Gamma$,
\be \sup_{t\in\mathbb R}\abs{\langle \delta_y, e^{-itH_\omega}\chi_{\mathcal I}(H_\omega)\delta_x\rangle} \leq \tilde C_{c,N_0,\zeta,\epsilon,\omega}e^{-(c-\epsilon)\textrm{dist}(x,y)^\zeta}.\ee
Now, suppose $E \in \cI$ is an eigenvalue of $H_\omega$. This eigenvalue is simple (see e.g. \cite[Theorem 5.8]{AW}), so we can use the spectral theorem to obtain the formula for the spectral projector on $E$,
\be P_{\{E\}}(H_\om)=\lim_{T\rightarrow \infty}\frac{1}{T}\int_0^Te^{isE}e^{-isH_\om}ds. \ee
Denote by $\varphi_{\om,E}$ the normalized eigenfunction corresponding to $E$. Then,
\begin{align}
\abs{\varphi_{\om,E}(x)} \abs{\varphi_{\om,E}(y)}& =\abs{ \angles{\delta_x,\varphi_{\om,E}}\angles{\varphi_{\om,E},\delta_y}  }\\ \notag
& =\abs{ \angles{\delta_x,\angles{\varphi_{\om,E},\cdot}\varphi_{\om,E},\delta_y}  }\\ \notag
& =\abs{\angles{\delta_x,P_{\{E\}}(H_\om)\chi_{\cI}(H_\om),\delta_y}}\\ \notag
& \leq \lim_{T\rightarrow \infty}\frac{1}{T}\int_0^T \abs{\angles{\delta_x,e^{-isH_\om}\chi_{\cI}(H_\om)\delta_y}  }ds\notag \\
& \leq \sup_{t\in\mathbb R} \abs{ \angles{\delta_x,e^{-itH_\om}\chi_{\cI}(H_\om)\delta_y}} \\
& \leq \tilde C_{c,N_0,\zeta,\epsilon,\omega}e^{-(c-\epsilon)\textrm{dist}(x,y)^\zeta}.
\end{align}
Here, we used  Fatou's lemma and the fact that $\angles{\varphi_{\omega,E},\cdot}=P_{\{E\}}(H_\om)=P_{\{E\}}(H_\om)\chi_I(H_\om)$ since $E\in \cI$. Since this holds for any $y$, we take $y$ to be a center of localization $x_0$, so $\varphi_{\om,E}(x_0)\neq 0$ (note $\varphi_{\om,E}(x_0)<\infty$ since $\varphi_{\om,E} \in\ell^2(\Gamma)$). Then
\be \abs{\varphi_{\om,E}(x)}  \leq \frac{1}{\abs{\varphi_{\om,E}(x_0)} }\tilde C_{c,N_0,\zeta,\epsilon,\omega}e^{-(c-\epsilon)\textrm{dist}(x,x_0)^\zeta} \ee
\end{proof}

From Theorems \ref{t:pp} and \ref{t:ef} we deduce the following
\begin{cor}Dynamical localization implies Anderson localization.
\end{cor}
The converse of this result is not true, as shown in \cite[Appendix 2]{dRJLS95}. There, the authors consider the quasiperiodic Schr\"odinger operator acting on $\ell^2(\NN\cup\{0\})$, $H=-\Delta+3\cos (\pi \alpha x+\theta)+\lambda \delta_x$, with $\theta,\lambda\in\RR$ and $\alpha$ irrational. For a particular choice of $\alpha$, they show that $H$ has pure point spectrum with exponentially decaying eigenfunctions $(\varphi_n)_{n\in\NN}$,
\be\label{eq:cn} \abs{\varphi_n(x)}\leq C_n e^{-\beta\norm{x}},\,\text{ for some constants}\,\beta,\,C_n>0. \ee
Therefore $H$ exhibits Anderson localization, however, it does not exhibit dynamical localization. Indeed, the authors show
\be\label{eq:prop} \lim_{t\rightarrow\infty} \frac{\norm{x{e^{-itH}}\delta_0}^2\ln t}{t^2}=\infty. \ee
This means that the quantity $\norm{x{e^{-itH}}\delta_0}^2$ is not bounded in $t$. What fails in their example is the lack of control in the constants $C_n$ in \eqref{eq:cn}. If the constant $C_n$ grows with the labelling $n$, the eigenfunction can become more and more spread, which might contribute to the type of propagation seen in \eqref{eq:prop}. The authors in \cite{dRJLS95} propose a stronger notion of eigenfunction localization called SULE (semi-uniform localized eigenfunctions) that avoids these pathologies. This property is actually obtained from the definition of dynamical localization \eqref{eq:dl}, see \cite[Section 7]{dRJLS95}. For a detailed discussion of the relation between dynamical localization and the decay of eigenfunctions, see \cite{GT}, which extends these notions to graphs $\Gamma$ with a certain condition on the growth of the volume of balls \cite[Theorem 2.9]{GT}. Namely, the volume of the balls in the graph distance should grow at most sub-exponentially with the radius.

\begin{figure}[h]
  \centering
  \includegraphics[width=14cm]{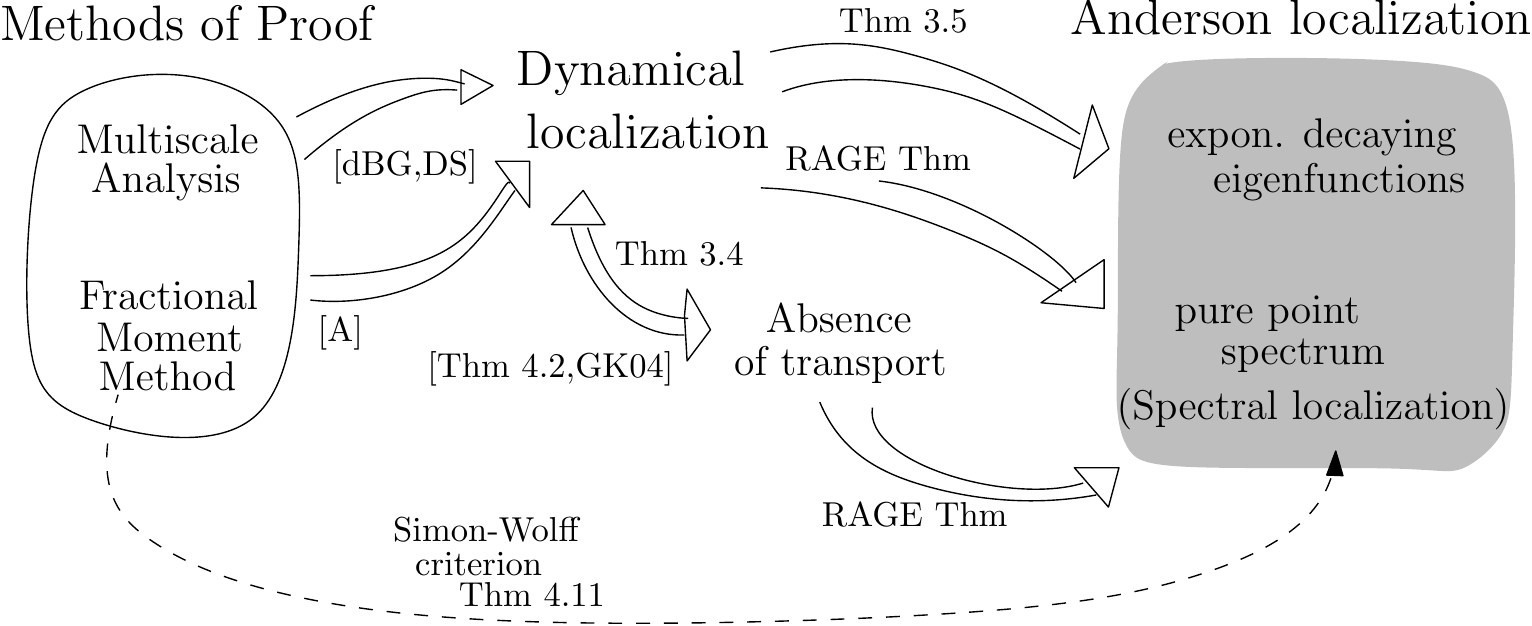}\\
  \caption{Summary of types of localization. The methods of proof will be discussed in Section \ref{s:dgf}.}
  %\label{}
\end{figure}

\section{The decay of Green's function and localization}\label{s:dgf}
We saw in the previous section that the decay of the term $e^{-itH_\om}$  in  \eqref{eq:dl} is crucial to prove localization in all its forms. Instead of studying this unitary operator directly, in our setting it is useful to study the resolvent of $H_\om$, also called the Green's function. Using the spectral theorem (see, for ex., \cite{RSI} or \cite{Go}) one can see that both are related by the formal identity
\be (H-z)^{-1}=i\int_0^\infty e^{-it(H-z)}dt,\quad z\in\CC,\,\textrm{Im}\,z>0. \ee
The Green's function is bounded outside the spectrum of $H_\om$ and its decay inside the spectrum gives information on the decay of $e^{-itH_\om}$, and actually, the decay the class of complex-valued measurable functions $f(H_\om)$ with $\norm{f}_\infty\leq 1$. Not only that, the behavior of $(H-z)^{-1}$ when $\textrm{Im}\,z\rightarrow 0$ also characterizes the spectral measures, see \cite{Si}. Therefore, it is not surprising that the existing methods to prove localization focus on obtaining estimates on the terms $\angles{\delta_x, (H-z)^{-1}\delta_y}$, $x,y\in\Gamma$.

In this section, we limit ourselves to describe briefly the methods available to obtain the decay of Green's function that implies dynamical localization in arbitrary dimension. For full proofs, we refer the reader to the introductory notes \cite{K,Kl}, plus the book \cite{Sto} on the Multiscale Analysis, and to \cite{S10,H08}, plus the book \cite{AW} on the Fractional Moment method.

\subsection{The Multiscale Analysis (MSA)}

Consider $\Gamma=\ZZ^d$, and denote by $\norm{x}$ the sup-norm in $\ZZ^d$. For a given bounded set $\L\subset \ZZ^d$, we will write $\Lambda^c=\ZZ^d\setminus \L$. We define its boundary,
inner and outer boundary, respectively, by
\be \partial\L=\{(u,v)\in\L\times\L^c;u\in\L\,,v\in\L^c\},\ee
\be \partial_+ \L=\{ v\in\L^c;\exists u\in \L \,\text{such that}\,(u,v)\in\partial\L \}, \ee
\be \partial_- \L=\{ u\in\L;\exists v\in \L^c \,\text{such that}\,(u,v)\in\partial\L \}. \ee
This implies the following decomposition for $H_{\om}$:
\be H_\om=H_{\om,\L}\oplus H_{\om,\L^c}+\Upsilon_\L, \ee
where
\be \angles{\delta_x,H_{\om,\L}\oplus H_{\om,\L^c}\delta_y}=\begin{cases} \angles{\delta_x,H_{\om,\L}\delta_y},\quad \text{if }x,y\in\L\\\angles{\delta_x,H_{\om,\L^c}\delta_y},\quad \text{if }x,y\in\L^c\\ 0\quad \mbox{otherwise} \end{cases},  \ee
and the boundary operator $\Upsilon_\L$ is given by
\be \angles{\delta_x,\Upsilon_\L\delta_y}= \begin{cases}-1,\quad \text{if }(x,y)\in\partial\L\\ 0\quad \mbox{otherwise}\end{cases}.\ee

We say that $\psi$ is a generalized eigenfunction of $H_\om$ with generalized eigenvalue $E$ if $\angles{\varphi,H_{\om}\psi}=E\angles{\varphi,\psi}$ for all $\varphi\in\ell^2_c(\ZZ^d)$. The following key observation enables us to obtain the decay of the generalized eigenfunctions from the decay of the Green's function: for any generalized eigenfunction $\psi$ with generalized eigenvalue $E$,
\be (H_{\om,\L}\oplus H_{\om,\L^c}-E)\psi=-\Upsilon_\L \psi.
\ee
Therefore, for $x\in\L$ we have
\begin{align}\psi(x)& =- \left(( H_{\om,\L}-E )^{-1} \Upsilon_\L\psi\right)(x)\\
%& =-\sum_{(k,m)\in\partial\L,k\in\partial_-\L,m\in\partial\L_+}\angles{\delta_x,( H_{\om,\L}-E )^{-1}\delta_k}\psi(m)
&=-\sum_{\substack{(k,m)\in\partial\L, \\ k\in\partial_-\L, \,m\in\partial\L_+}    } G_{\om,\L}(E;x,k)\psi(m),
\end{align}
where we write $G_{\om,\L}(E;x,k)=\angles{\delta_x,( H_{\om,\L}-E )^{-1}\delta_k}$.

We consider the finite-volume restriction of $H_\om$ to $\L$, denoted by $H_{\om,\L}$, by taking the restriction of $\chi_{\L}H_\om\chi_{\L}$ to $\ell^2(\L)$. We obtain thus a finite-volume operator (a finite matrix) and therefore its spectrum is discrete. The goal is to prove the decay of the terms $\angles{\delta_x,(H_{\om,\L}-E)^{-1}\delta_y}$ when $x$ and $y$ are distant, for an increasing sequence of sets $\L$ which exhausts the whole space, see Fig. \ref{box}. The problem that appears here is that since $E\in\sigma(H_\om)$ and the spectrum $\sigma(H_{\om,\L})$ is random, we might have that $\sigma(H_{\om,\L})$ is arbitrarily close to $E$, the quantity $(H_{\om,\L}-E)^{-1}$ being unbounded. To control the appearance of singularities we exploit the fact that $\sigma(H_{\om,\L})$ ''moves with the randomness''. This is known as the Wegner estimate.
\begin{defn}\label{d:we}Let $I\subset\mathbb R$ be an open interval. We say the operator $H_\om$ satisfies a Wegner estimate in $I$ if there exists a finite constant $Q_{I}$ such that
\be\label{eq:we} \mathbb P\left(\norm{(H_{\om,\L}-E)^{-1}}\geq \frac{1}{\eta}\right)\leq Q_I\eta^a \rm{vol}(\L)^b,\ee
for all $E\in I$, $\eta\in(0,1]$, some $b\geq1$, $a>0$ and all $L\in\NN$ large enough.
\end{defn}
For a proof, see for ex. \cite[Section 5.5]{K} or \cite[Chapter 4]{AW}.
\begin{rem}\begin{itemize}
\item[i.]The Wegner estimate, which usually holds throughout the spectrum, is a consequence of the regularity of the probability distribution $\mu$, which needs to be at least H\"older continuous. It can also be obtained for more singular distributions, like Bernoulli, but the bound in \eqref{eq:we} is not good enough to perform the MSA in dimension $d>1$. For references see the historical notes in \cite[Section 5.5]{K} or \cite[Chapter 4]{AW}.
\item[ii.] For the Anderson model on $\ell^2(\mathbb Z^d)$, \eqref{eq:we} is known to hold with $a=b=1$, see e.g. \cite[Section 5.5.]{K}. For different models the Wegner estimate can be obtained with different values of $a$ and $b$, see the discussion in \cite[Remark 4.6]{Kl}.
    \end{itemize}
\end{rem}
Next, we need a quantitative estimate on the desired decay of the Green's function.
From now on, we let $\L_L(u)$ be the box of side $L$ centered on $u\in\ZZ^d$, given by
\be \L_L(u)=\left\{ v\in\ZZ^d;\norm{u-v}\leq \frac{L}{2}\right\}, \ee

\begin{defn}\label{d:gb} We say that the box $\L_L(u)$ is $(m,E)$-good if $E\notin \sigma(H_{\om,\L})$ and
\be\label{eq:gb} \abs{\angles{\delta_x,(H_{\om,\L}-E)^{-1}\delta_y}}\leq e^{-m \frac{L}{2}}, \ee
for any $x\in \L_{L/2}(u)$ and $y\in\partial_-\L_L(u)$.
\end{defn}
We define
\be\label{eq:pel} p_{E,L,m,u}:=\P(\L_{L}(u) \,\text{is not}\,(m,E)\text{-good}). \ee
The desired decay of the resolvent at the scale $L$ becomes the statement
\be\label{eq:pelk} p_{E,L,m,u}\leq \frac{1}{L^{\beta}},\quad\text{for some}\,\beta>0. \ee
This means that for a large set of $\om\in\Omega$, the resolvent decays as \eqref{eq:gb} between points $x,y$ that are at a distance proportional to $L$.

\begin{figure}[h]
  \centering
  \includegraphics[width=4cm]{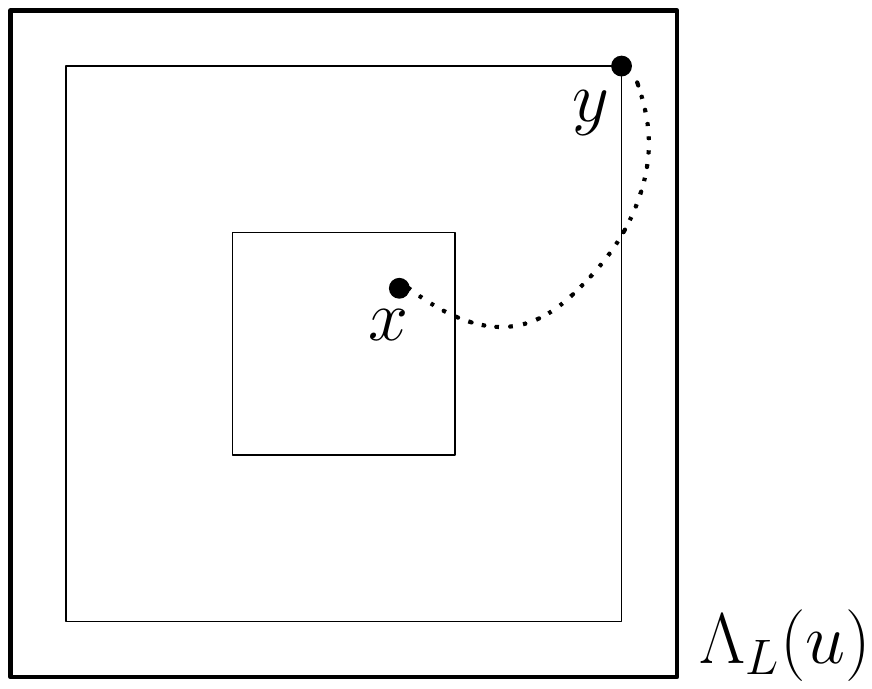}\\
  \caption{The goal is to study the decay of the resolvent from the core to the boundary of boxes of side-length $L$.}
  \label{box}
\end{figure}

\begin{rem}Note that when $H_\om$ is ergodic with respect to translations in $\ZZ^d$, $p_{E,L,m,u}=p_{E,L,m,0}$ for all $u\in\ZZ^d$. Therefore it is enough to consider estimes on a box $\L=\L_{L}(0)$. If $H_\om$ is not ergodic, like in the case of Delone-Anderson operators, in order to perform the MSA method we need estimates that are uniform with respect to the center $u\in\ZZ^d$ of the box, so condition \eqref{eq:pelk} becomes (see \cite{RM12})
\[\sup_{u\in\ZZ^d}p_{E,L,m,u}\leq \frac{1}{L^{\beta}},\quad\text{for some}\,\beta>0.\]
\end{rem}

The Multiscale Analysis method (MSA) consists of an induction on scales. In order to prove localization in a given interval $I\subset\RR$, one needs to prove the following
\begin{itemize}
\item[i.] The operator satisfies a Wegner estimate on $I$ of the form \eqref{eq:we}.
\item[ii.] The \emph{initial length scale estimate}: there exists a finite initial length-scale $L_0$ and $\beta>0$ such that for all $E\in I$, \eqref{eq:pelk} holds for $L_0$.
\item[iii.] The \emph{induction step}: if, for $E\in I$, \eqref{eq:pelk} holds at a scale $L_k$, then it holds at the scale $L_{k+1}=L_k^\alpha$, for a suitably chosen $\alpha\in(1,2)$.
\end{itemize}

There are several versions of this method in the literature, with different degrees of refinements. We state the single-energy multiscale analysis from \cite[Theorem 2.2]{vDK}, an improvement of the MSA \cite{FrS83},
\begin{thm}\label{vDK}Let $I\subset\RR$ be an interval. Suppose that
\begin{itemize}
\item[(H1)] $H_{\om}$ satisfies the Wegner estimate \eqref{eq:we} on $I$.
\item[(H2)] there exists a finite scale $L_0$ for which
\be p_{E,L_0,m_0,0}\leq \frac{1}{L_0^{\beta}},\quad\text{for some}\,\beta>2d,m_0>0. \ee
\end{itemize}
Then, there exists $\alpha\in(1,2)$ such that if we set $L_{k+1}=L_k^{\alpha}$, $k=0,1,2,...$ and pick $m\in(0,m_0)$, we can find $\mathcal L<\infty$ such that if $L_0>\mathcal L$, we have for all $k=0,1,2,...$
\be\label{eq:msa1} p_{E,L_k,m,0}\leq \frac{1}{L_k^\beta}. \ee
\end{thm}

We do not explain in detail the induction procedure, for this see \cite{K} or \cite{Kl}. We limit ourselves to sketch the idea in Figure \ref{f:indstep}.

\begin{figure}[h]
  \centering
  \includegraphics[width=12cm]{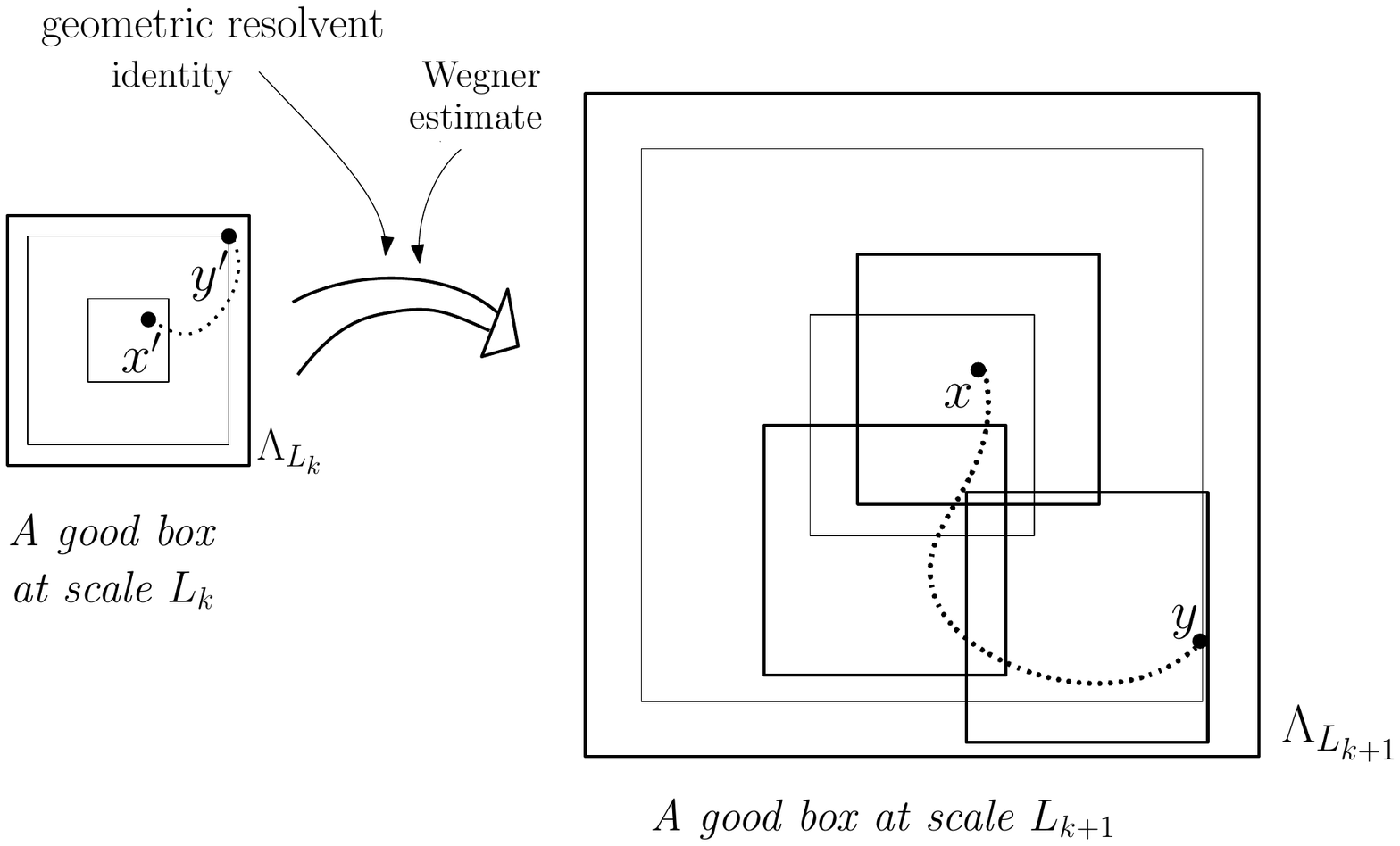}\\
  \caption{The induction step uses the information on the decay of the resolvent restricted to a box of side-length $L_k$ to show the decay of the resolvent at a scale $L_{k+1}$. The idea is to go from $x$ to $y$ in the cube $\Lambda_{L_{k+1}}$ via a path of overlapping boxes of sidelength $L_k$. Loosely speaking, using the geometric resolvent identity, the result is the product of the decay of the resolvent on the boxes of scale $L_k$. If a box $\Lambda_{L_k}$ is good, by definition it contributes a factor exponentially decaying in $L_k$. If a box $\Lambda_{L_k}$ is not good (bad), one can use the Wegner estimate to bound the contribution, which then adds a factor polynomial in $L_k$. If the conditions to perform the MSA hold, the exponential decay from good boxes dominates, i.e, compensates for the contributions of powers of $L_k$, giving the overall decay of the resolvent restricted to $\Lambda_{L_{k+1}}$.}\label{f:indstep}
\end{figure}

Assumption (H2) in Theorem \ref{vDK} corresponds to the initial length scale estimate. It can be usually shown either at the bottom of the spectrum, with $\lambda>0$ fixed, or in the whole spectrum, if $\lambda$ is large enough, see \cite[Section 11]{K} for details. Because of its uses in different settings, other than $\ell^2(\ZZ^d)$, we describe briefly the argument used for intervals $I$ that contain the spectral infimum $E_0=\inf\sigma(H_{\om,\lambda})$. One first needs to show that restricting the operator to a finite volume produces a gap in the spectrum, that is $E_0-E_{0,\L}=f(L)$ , where $E_{0,\Lambda}=\inf\sigma(H_{\om,\lambda,\L_L})$ and $f(L)$ is a function of $L$ that decays polynomially.

Once the spectral gap is proven, assumption (H2) is the consequence of the Combes-Thomas estimate, which asserts the decay of the term $G_{\om,\L_L}(E;x,y)$ in terms of $\dist(E,\sigma(H_{\om,\L_L}))$ and $\norm{x-y}$. We recall this result as stated in \cite[Appendix A]{KlNRM} for general finite-volumes (not necessarily a box). See \cite[Section 10.3]{AW} for other versions.

\begin{thm}[Combes-Thomas estimate]
Let $H = -\Delta + V$  be a Schr\"odinger operator on $\ell^{2}(\ZZ^d)$. Given $S \subset \ZZ^d$,  let $H_S$  be the  restriction of $\Chi_S H \Chi_S$ to $\ell^2(S)$.  Then for every $z \notin \sigma(H_S)$,  setting $\eta_z = \dist \pa{z, \sigma(H)}$, for all $\eps \in ( 0, 1)$ we get
\be\label{CTest}
\abs{ \angles{\delta_x, (H_S - z)^{-1} \delta_y}  }   \le \tfrac{1}{\eta_z (1-\epsilon)} e^{- \log  \left(\tfrac{\epsilon \eta_z}{2d} + 1\right)  \norm{y - x}}\,\text{for all}\,\, x,y \in \,S.
\ee
\end{thm}

From \eqref{CTest} we see that in order for the r.h.s. to decay in $L$, we need to have a gap size $f(L)$ to be not too small. The existence of this gap at the bottom of the spectrum can be shown using the Lifshitz tails behavior of the Integrated Density of States (IDS), when the latter exists (see \cite[Section 6 and Section 11.3]{K}). In the case where the IDS is not available, other arguments can be used to show that finite-volume restriction "lifts up" the spectrum, for example, a space averaging argument (see \cite{RM13} and references therein).

  \bigskip

The output of the MSA, \eqref{eq:msa1} for a sequence of scales $L_k$, implies the absence of diffusion \cite{FrS83}, and absence of absolutely continuous spectrum in $I$ \cite{MS}, but it is not enough to conclude dynamical localization. For this, a stronger version of the MSA is needed, which adds to the induction step an analysis on pairs of cubes, instead of only one cube. When working on the lattice $\ZZ^d$, the following is a consequence of the independence of the random variables,
\begin{defn}\label{d:web}Let $I\subset \RR$ be an open interval. We say the operator $H_\om$ satisfies a Wegner estimate between boxes if, for any two disjoint boxes $\L_1$ and $\L_2$ there exists a finite constant $Q_J$ for each compact interval $J\subset I$ such that
\begin{align}\label{eq:web} & \P\left(\exists E\,\text{such that} \dist(E,\sigma(H_{\om,\L_1}))<\eta\,\,\text{and}\, \dist(E,\sigma(H_{\om,\L_2}))<\eta\right) \notag\\
& \quad \quad \quad \quad \quad \quad \quad \quad \quad \quad \quad \quad \quad \quad \quad  \leq Q_I \eta\,\textrm{vol}(\L_1)\textrm{vol}(\L_2).
\end{align}
\end{defn}
We now state the energy-interval MSA from \cite[Theorem 2.2]{vDK}.
\begin{thm}\label{t:vdk}Let $I\subset\RR$ be an interval. Suppose that
\begin{itemize}
\item[(H1)'] $H_{\om}$ satisfies both the Wegner estimate \eqref{eq:we}  and the Wegner estimate between boxes \eqref{eq:web} in $I$.
\item[(H2)] there exists a finite scale $L_0$ for which
\be p_{E,L_0,m_0,0}\leq \frac{1}{L_0^{\beta}},\quad\text{for some}\,\beta>2d,m_0>0. \ee
\end{itemize}
Then, there exists $\alpha\in(1,2)$ such that if we set $L_{k+1}=L_k^{\alpha}$, $k=0,1,2,...$ and pick $m\in(0,m_0)$, we can find $\mathcal L<\infty$ such that if $L_0>\mathcal L$, we have for all $k=0,1,2,...$
\be\label{outMSA} \P\left(\forall E\in I,\, \text{either}\,\L_{L_k}(x) \,\text{or}\,\L_{L_k}(y)\,\text{is }(m,E)\text{-good}\right)\leq \frac{1}{L_k^{\beta}}. \ee
\end{thm}

The output of the energy-interval MSA \eqref{outMSA} implies dynamical localization in the sense of \eqref{eq:dl}, with an exponent $\zeta<1$, and in the sense of \eqref{eq:at}, see \cite{dBG,DamS}. In \cite{dBG} the authors use ideas from \cite{dRJLS95} and the notion of semi-uniform localized eigenfunctions (SULE). They show that the output of the MSA \cite{vDK}, originally used to prove Anderson localization, can be exploited further to obtain a version of SULE, that gives exponential decay of eigenfunctions, and therefore dynamical localization in the sense of absence of transport \eqref{eq:at}, see Remark \ref{r:loc}-(i).

A pedagogical explanation of how the output of the MSA implies Anderson localization is given in \cite[Section 9]{K}.

We summarize the localization results for the particular case of the Anderson model on the square lattice. In dimension $d\geq1$ we have:
\begin{thm}\label{t:andloc}Let $H_{\om,\lambda}=-\Delta+\lambda V_\om$, with $\lambda>0$, be the Anderson model on $\ell^2(\ZZ^d)$ defined in \eqref{eq:and} and let $\Sigma_\lambda\subset\RR$ be its deterministic spectrum. Assume the probability distribution $\mu$ is H\"older continuous and that $\supp \mu=[a,b]$, with $a<b$. Let $I\subset \RR$ be an open interval. Then,
\begin{itemize}
\item[i.] (regime of extreme energies) for fixed $\lambda>0$, there exist $E_1, E_2\in\RR$ depending on $d,\lambda,\mu$, such that $H_{\om,\lambda}$ exhibits dynamical localization in $I$ if
\be I\subset [\inf \Sigma_\lambda,E_1]\cup [E_2,\sup \Sigma_\lambda],\ee
\item[ii.](regime of high disorder) if $\lambda>0$ is large enough depending on $\mu$ and $d$, then $H_{\om,\lambda}$ exhibits dynamical localization in any interval $I\subset\Sigma_\lambda$.
\end{itemize}
In particular, in both cases $H_{\om,\lambda}$ exhibits all forms of localization stated in Definition \ref{def:loc}.
\end{thm}
\begin{rem}In the particular case $d=1$, $\lambda>0$ is enough for $H_{\om,\lambda}$ to exhibit dynamical localization throughout the spectrum, that is, for any $I\subset \Sigma_\lambda$, see e.g. \cite{GMP77,Pas80, Kot}.
\end{rem}

\subsubsection{The case of general metric graphs $\Gamma$}
 In the MSA, the induction step is carried out by covering a cube of side $L_{k+1}$ by smaller cubes of side $L_{k}$. Cubes are the natural choice to tile the square lattice, although the method can be applied to more general $\Gamma$ and more general geometric shapes that tile the space. One can consider the Anderson model with missing sites from Section \ref{s:ne}, where the potential is forced to be zero in an homogenous subset of $\ZZ^d$, and use the MSA to prove localization \cite{RM13,EK} (see \cite{ES} for a proof of localization using the Fractional Moment Method). One can also consider $\ell^2(\Gamma)$, with $\Gamma$ satisfying Assumption (A1) in Section \ref{s:and} provided the volume of balls of radius $L_k$ in the graph-distance grows at most polynomially in $L_k$ as $L_k\nearrow \infty$, see \cite{ChuS}. This fact can be seen from the induction step, in the way \emph{bad} regions are handled: a bad region is a collection of cubes of side $L_*\geq L_{k}$, that contains the collection of cubes $\L_{L_{k}}(j)\subset \L_{L_{k+1}}$ that are not good in the sense of Definition \ref{d:gb}. When the path between $x\in\L_{L_{k+1}/2}$ and $y\in\partial_-\L_{L_{k+1}}$ passes through one of the bad regions, its contribution to the decay of the Green's function is given by the Wegner estimate. This term will be therefore proportional to the volume of the bad region, but should be compensated by the exponential decay of the Green's function in the \emph{good} cubes. This is the reason why, in particular, the MSA cannot be applied in trees, like the Bethe lattice $\mathbb B$, where the volume of cubes grows exponentially.

% {{ Multiscale Analysis (MSA)}}\\
% (Weak version) Prove that for some interval $I\subset\mathbb R$ the following holds: for some $\alpha>1$, $p>2d$ and $\gamma>0$ and for all $E\in I\subset\mathbb R$, there is a sequence of cubes $\Lambda_{L_k}$, $L_{k+1}=L_k^\alpha$, ${L_k} \nearrow\mathbb Z^d$,
% \[ \mathbb P\left( \abs{\langle \delta_x,(H_{{\omega,\lambda}}\restriction_{\Lambda_{L_k}}-E)^{-1}\delta_y\rangle}\leq e^{-\gamma L_k} \right)\geq 1- \frac{1}{L_k^p}. \]

\subsection{The Fractional Moment Method (FMM)}
We consider the Anderson operator $H_{\om,\lambda}$ on $\ell^2(\Gamma)$ defined in \eqref{eq:and}, with $\Gamma$ satisfying Assumption (A1) and the random variables distributed according to an absolutely continuous probability distribution $\mu$ .

The Fractional Moment Method, when applicable, gives exponential decay of fractional powers of the resolvent $G(z;x,y)$, uniformly on $z$ with $\Im z\neq 0$.
\begin{thm}\label{t:fmm}
Let $H_{\om,\lambda}=-\Delta+\lambda V_\om$ be the Anderson model and assume $\lambda$ is large enough. For some $I\subset\mathbb R$, there exist $s\in(0,1)$ and $0<c$, $C<\infty$ such that
\be\label{eq:fmm}\mathbb E\left( \abs{ \langle\delta_x,(H_{\omega,\lambda}-(E+i\epsilon))^{-1}\delta_y \rangle}^s\right)\leq Ce^{-c\dist_\Gamma(x,y)}\ee
uniformly in $E\in I$, $\epsilon>0$ and $x,y\in\mathbb Z^d$.
\end{thm}
This analysis can be done directly on the infinite-volume operator, without restricting it to finite-volume sets. Just as in the MSA method the singularities of the resolvent are controlled by the Wegner estimate, in the Fractional Moment method this is done by taking fractional powers of the resolvent. In order to apply the FMM one needs to verify the following,

\begin{itemize}
\item[1)] An a priori bound on fractional moments: There exists a constant $C_1=C_1(s,\rho)<\infty$ such that for all $x,y\in\Gamma$,
    \be\label{eq:ap} \mathbb E\left( \abs{G_{\om,\lambda}(z;x,y)}^s\right)\leq \frac{C_1}{\lambda^s}. \ee
See \cite[Lemma 4.1]{S10}, \cite[Lemma 5]{Graf} for a proof in $\ZZ^d$ and \cite[Lemma 3.1]{T11} for a proof in a general graph $\Gamma$. We describe briefly the steps that lead to \eqref{eq:ap} below.

\item[2)] A decoupling estimate for the probability density of the random variables: There exists a constant $C_2=C_2(s)<\infty$ such that, uniformly in $\alpha,\beta\in\CC$,
    \be\label{eq:decl} \int\frac{1}{\abs{v-\beta}^s}\rho(v)dv\leq C_2 \int \frac{\abs{v-\alpha}^s}{\abs{v-\beta}^s}\rho(v)dv. \ee
      This lemma holds for probability densities that are piecewise continuous, and with enough decay at infinity. This restriction can be weakened, allowing for probability distributions $\mu$ that are not necessarily absolutely continuous, but with a certain degree of regularity, see \cite[Section 3.1]{AM93}.
    \end{itemize}
The proof of \eqref{eq:ap} consists in being able to write the term $G_\om(z;x,y)$ in such a way that the potential at the sites $x$ and $y$ is seen as a finite-rank perturbation. For the diagonal terms $G_\om(z;x,x)$ this is easily done by considering the potential at site $x$ as a rank one perturbation of the operator:
\be H_{\om,\lambda}=H_{\om_x^\bot,\lambda}+\lambda \om_x P_x, \ee
where we write $\om_x^\bot=(\om_u)_{u\neq x}$, and $P_x=\angles{\delta_x,\cdot}\delta_x$ is the rank one projection on the site $x\in\Gamma$.
To compute $G_\om(z;x,x)$ we use the resolvent identity: for self-adjoint operators $A,B$,
\be A^{-1}-B^{-1}=A^{-1}(B-A)B^{-1}, \ee
(see e.g. \cite{RSI}) which gives,
\be G_{\om}(z;x,x)=G_{\om_x^\bot}(z;x,x)-\lambda \om_x G_{\om_x^\bot}(z;x,x)G_{\om}(z;x,x). \ee
From this, one can deduce that for some $a,b\in\CC$, depending on the value of the random potential on all sites except for $\om_x$,
\be G_{\om}(z;x,x)=\frac{a}{b+\lambda\om_x}. \ee
The advantage of isolating the value of the potential at $x$ is that now we can take the expectation with respect to $\om_x$ only, $\mathbb E_x$,
%we get
%\be \mathbb E_x\left(G_{\om}(z;x,x)\right)\leq \int\int\frac{a\rho(v)}{b+\lambda\om_x}dv. \ee
 then \eqref{eq:ap} for $x=y$ is a consequence of the following fact that holds for compactly supported random variables (see  \cite[Section 2]{H08}),
\be\label{eq:reg} \int \frac{\rho(v)}{\abs{\beta-v}^s}dv<C(s,\rho)<\infty,\quad \forall \beta\in\CC. \ee
The fact that $s<1$ is crucial for this integral to be finite, see \cite[Section 2]{H08}.

To bound the off-diagonal terms $x\neq y$, one proceeds in a similar way, but this time considering the potential at the sites $x$ and $y$ as a rank-two perturbation of the operator,
\be H_{\om,\lambda}=H_{\om_{x,y}^\bot,\lambda}+\lambda \om_x P_x+\lambda \om_y P_y, \ee
where we write $\om_{x,y}^\bot=(\om_u)_{u\notin \{x,y\}}$. In order to compute the resolvent, one needs a particular case of the Krein formula  for a projector of rank two. We recall this formula from \cite[Appendix I]{AM93},

\begin{thm}[Krein formula]
Let $H$ be a self-adjoint operator on some Hilbert space $\mathcal H$. If
\be H=H_0+W,\ee
with $W$ a finite rank operator satisfying
\be W=PWP \ee
for some finite-dimensional orthogonal projection $P$, then, for $z$ with ${\rm Im}z\neq 0$, we have
\be \left[ P(H-z)^{-1}P\right]=\left[ W + \left[ P(H_0-z)^{-1}P\right]^{-1} \right]^{-1}  \ee
where the inverse is taken on the restriction to the range of $P$.
\end{thm}

From the definition of the resolvent, one can deduce
\be (\lambda\om_y-z)G_{\om}(z;x,y)=\sum_{u:u\sim y}G_\om(z;x,u), \ee
then use \eqref{eq:decl} to obtain (see \cite[Eq. 28,29]{S10})
%the exponential decay of the term $\mathbb E\left( \abs{G_\om(z;x,y)}^s \right)$ is the result of iterating the a priori bound, see \cite[Equation 29]{S10},
\be \mathbb E\left(\abs{G_\om(z;x,y)}^s \right) \leq \frac{C_1}{\lambda^s}\sum_{u:u\sim y} \mathbb E\left(\abs{G_\om(z;x,u)}^s \right). \ee
and iterate this estimate as long as $u\neq x$, obtaining
\begin{align} \mathbb E\left(\abs{G_\om(z;x,y)}^s \right) & \leq \left( \frac{KC_2}{\lambda^s}\right)^{\textrm{dist}_\Gamma(x,y)}\sup_{u\in \ZZ^d}\mathbb E\left(\abs{G_\om(z;x,u)}^s \right) \notag\\
& \leq  \left( \frac{KC_2}{\lambda^s}\right)^{\textrm{dist}_\Gamma(x,y)+1},
\end{align}
where in the last step \eqref{eq:ap} was used.
Taking $\lambda$ large enough, the r.h.s. decays exponentially in $\textrm{dist}_\Gamma(x,y)$, proving Theorem \ref{t:fmm}.

An alternative method to deduce the exponential decay \eqref{eq:fmm} was introduced by Hundertmark in \cite{H08}, using the self-avoiding random walk representation for the finite-volume Green's function. We describe briefly his approach.
A self-avoiding random walk is a random sequence of points, or path $\{w_0,w_1,...\}\subset \Gamma$ where $w_n$ and $w_{n+1}$ are nearest neighbors, for $n\geq 0$, and all points $w_n$ are different.
We start by recalling  \cite[Lemma 4.3]{H08},

\begin{lem} Let $\L\subset\Gamma$ be finite and $w=\{w_0,...,w_{\abs{w}}\}$ denote the  self-avoiding random walk (SAW) connecting $x$ and $y$, where $\abs{w}$ is the length of the walk. We define the sets $\L_j$  by
\be\label{eq:cub} \L_0=\L,\quad \L_{j+1}=\L_{j}\setminus\{w_j\},\quad j=0,1,2,... \ee
Then, the Green's function restricted to $\L$ takes the following form
\be\label{eq:saw} G_{\om,\L}(z;x,y)=\sum_{w\in SAW_\L(x,y)} \prod_{j=0}^{\abs{w}}G_{\om,\L_j}(z;w_j,w_j) ,\ee
where $SAW_\L(x,y)$ denotes all the SAW between $x$ and $y$ contained in $\L$.
\end{lem}

The previous statement gives a bound on the infinite-volume Green's function via the relation
\be \mathbb E\left(\abs{G_\om(z;x,y)}^s \right)\leq \liminf_{\L\nearrow \Gamma} \mathbb E\left(\abs{G_{\om,\L}(z;x,y)}^s \right) \ee

This lemma avoids the use of the decoupling estimate \eqref{eq:decl}, and uses only the a priori bound for diagonal terms \eqref{eq:ap}, which relies on \eqref{eq:reg}.
It is obtained from expanding the resolvent as a Neumann series in terms of powers of the Laplacian $\Delta$. The fact that $\angles{\delta_u,\Delta \delta_v}\neq 0$ only when $u\sim v$ implies that the non zero terms in the sum come from a path of nearest neighbors connecting $x$ and $y$. Taking all possible paths between $x$ and $y$ can be re-arranged and be represented as a self-avoiding random walk.

Again, this approach manages to isolate the value of the potential at a site $w_j$, that is, each resolvent in the r.h.s. of \eqref{eq:saw} depends only on the value of the potential at site $w_j$. The independence of the random variables $\om_{w_j}$ allows for a simple computation, where the  expectation decouples in single-site contributions:
\begin{align} & \mathbb E\left(\abs{G_{\om,\L}(z;x,y)}^s\right)\\
 & \quad\quad\quad \leq \sum_{w\in SAW_\L(x,y)} \mathbb E_{\om_{w_0}^\bot}\left( \prod_{j=1}^{\abs{w}} G_{\om,\L_j}(z;w_j,w_j)\right)
\mathbb E_{\om_{w_0}} \left(G_{\om,\L_j}(z;w_0,w_0)\right) \notag\\
&\quad\quad\quad \leq  \frac{C_1}{\lambda^s}\sum_{w\in SAW_\L(x,y)} \mathbb E_{\om_{w_{0},w_1}^\bot}\left( \prod_{j=2}^{\abs{w}} G_{\om,\L_j}(z;w_j,w_j)\right)
\mathbb E_{\om_{w_1}} \left(G_{\om,\L_j}(z;w_1,w_1)\right) \notag\\
&\quad\quad\quad ... \notag\\
&\quad\quad\quad \leq \sum_{w\in SAW_\L(x,y)} \left(\frac{C_1}{\lambda^s}\right)^{\abs{w}+1}
 \end{align}
 where we used the a priori bound \eqref{eq:ap}. We can bound the last line by a sum that considers all possible self-avoiding random walks between $x$ and $y$ not necessarily restricted to $\L$, which we denote by $SAW(x,y)$, and get
 \be\mathbb E\left(\abs{G_{\om,\L}(z;x,y)}^s\right) \leq  \sum_{w:SAW(x,y)} \left(\frac{C_1}{\lambda^s}\right)^{\abs{w}+1}. \ee
  Note that the length of the self-avoiding random walk $\abs{w}$ is at least $\textrm{dist}_\Gamma(x,y)$. Therefore, we can write
  \be\mathbb E\left(\abs{G_{\om,\L}(z;x,y)}^s\right) \leq  \sum_{n\geq \dist_\Gamma(x,y)} \left(\frac{C_1}{\lambda^s}\right)^{\abs{n}+1}, \ee
 which is bounded if $\lambda$ is large enough and the graph $\Gamma$ is such that the number of self-avoiding random walks of length $n$ grows less than exponentially in $n$. From this one obtains Theorem \ref{t:fmm}.

The self-avoiding walk representation shows more directly the structure of the graph $\Gamma$, and allows for graphs with unbounded vertex degree \cite{T11}. Using the SAW representation, J. Schenker was able to give a rigourous proof of the critcal value $\lambda_*$ of the disorder parameter conjectured by P.W. Anderson for which there is localization for all $\lambda\geq \lambda_*$ \cite{Sche}.

So far, the arguments relied on $\lambda$ being very large. One can also treat the energy regime near the spectral band edges, with fixed $\lambda>0$, restricting the operator to finite volumes, see \cite[Section 7]{S10}. Thus, one retrieves Theorem \ref{t:andloc} with the condition on the probability distribution $\mu$ replaced by a condition such that \eqref{eq:decl} holds.

In their original work \cite{AM93}, the authors proved spectral localization directly from \eqref{eq:fmm} using the Simon-Wolff Criterion \cite{SW}, which we state as in \cite[Theorem 5.7]{AW},
\begin{thm}[The Simon-Wolff Criterion]\label{t:sw} Let $\Gamma$ be a countable set of points. Let $H_{\omega}=-\Delta+V_\omega$ on $\ell^2(\Gamma)$, $\om\in\Omega$, such that the probability distribution of the random variables, $\mu$, is absolutely continuous. Then,
for any Borel set $I$:
\begin{itemize}
\item[i.] If for Lebesgue-a.e. $E\in I$ and $\mathbb P$-a.e. $\omega$
\be\label{eq:pp} \lim_{\epsilon\rightarrow 0} \sum_{y\in\Gamma } \abs{\langle \delta_y, (H_{\omega}-(E+i\epsilon))^{-1}\delta_x\rangle}^2<\infty, \ee
then for $\mathbb P$-a.e. $\omega$,  the spectral measure of $H$ associated to $\delta_x$ is pure point in $I$.
\item[ii.] If for Lebesgue-a.e. $E\in I$ and $\mathbb P$-a.e. $\omega$
\be \lim_{\epsilon\rightarrow 0} \sum_{y\in\Gamma } \abs{\langle \delta_y, (H_{\omega}-(E+i\epsilon))^{-1}\delta_x\rangle}^2=\infty, \ee
then for $\mathbb P$-a.e. $\omega$,  the spectral measure of $H$ associated to $\delta_x$ is continuous in $I$.
\end{itemize}
\end{thm}
We will use the following inequality, valid for $s\in(0,1)$ and a sequence $(a_n)_{n\in\NN}\subset \CC$,
\be \left(\sum_n \abs{a_n}\right)^s\leq \sum_n \abs{a_n}^s.\ee
For $s=1/4$ we have,
\be\left(\sum_y  \abs{ \langle\delta_x,(H_{\omega,\lambda}-z)^{-1}\delta_y \rangle}^2\right)^{\frac{1}{4}}\leq
\sum_y \abs{ \langle\delta_x,(H_{\omega,\lambda}-z)^{-1}\delta_y \rangle}^{\frac{1}{2}}<\infty\ee
for $\mathbb P$-a.e. $\omega\in\Omega$. Therefore, Theorem \ref{t:fmm} implies \eqref{eq:pp} for any $x\in\Gamma$, and by the Simon-Wolff Criterion, the spectral measure associated to $H_\omega$ and $\delta_x$ is pure point in the deterministic spectrum of $H_\omega$, for $\mathbb P$-a.e. $\omega\in\Omega$. Since this holds for every $\delta_x$, one can deduce that
\be\sigma(H_\omega)=\sigma_{pp}(H_\omega) \quad\mbox{for }\mathbb P\mbox{-a.e.}\,\omega\in\Omega. \ee

%\subsection{
%Results for the Anderson model on graphs (ex. $\ell^2(\mathbb B)$)}
%\smallskip
%\begin{itemize}
%\item[$\bullet$] Delocalization and ac spectrum, $\ell^2(\mathbb B)$\\
%{\small Klein '96- '98, Aizenman-Sims-Warzel'06, Froese-Hasler-Spitzer'06,'07, Halasan'09, Aizenman-Warzel'06--'16.}
%\item[$\bullet$] Integrated Density of States.\\
%{\small Acosta-Klein'92, Hoecker--Escuti-Schumacher'12 ($\mathbb B$), Antunovi\'c-Veseli\'c'08}
%\end{itemize}
%\bigskip

The Simon-Wolff criterion is based on the theory of rank-one perturbations, which are at the core of many of the arguments in the spectral analysis of Schr\"odinger operators, and in particular in the early proofs of localization, see \cite{Si,dRJLS95}. The disadvantage is the requirement of absolutely continuous probability distributions.

In \cite{A} it was shown that the exponential decay of Green's function obtained by the FMM implies dynamical localization. This was proved in more generality by \cite{Graf}, see also \cite[Proposition 5.1]{S10}, \cite[Lemma 4.3 and Section 5]{T11}.

\subsection{A tale of two methods: the MSA and the FMM}
A brief comment is in place to compare the two methods presented above. These are the two methods available in arbitrary dimension to prove dynamical localization, and given a random Schr\"odinger operator, the question is which one to choose. They both have advantages and disadvantages: the MSA is a more cumbersome method than the FMM, but once one learns it in the discrete setting $\ell^2(\mathbb Z^d)$, the passage to the continuous setting $\Lp{\mathbb R^d}$ is straightforward. The FMM, on the other hand, becomes much more technical in the continuous setting \cite{AENSS}, since the finite-rank arguments at its core are no longer available.
Both methods give exponential decay in \eqref{eq:dl}, i.e., $\zeta=1$ there. This is is part of the proof of localization with the FMM case. In the proofs using MSA, the sub-exponential decay can be computed directly, and for many years, it was believed this was the optimal result. In the more refined study \cite{GKnew} it was shown that actually exponential decay holds, see also \cite[Remark 1.7]{GK11}. Both methods, MSA and FMM, yield localization in the form \eqref{eq:at}.

While for an Anderson operator with random variables that have absolutely continuous probability distributions both methods can be applied, the MSA has proved to give the best adaptability to singular probability measures. The proof of localization for the Anderson model with Bernoulli random variables was obtained in \cite{BoK05,GK11} using a version of the MSA method, where the Wegner estimate to control resonances is replaced by a weaker estimate that is incorporated in the induction step.

We can see in the induction step of the MSA that bad boxes are allowed, as long as they are compensated by the exponential decay of the Green's function on good boxes. Therefore it is important that the volume of a cube $\L_L$  in the graph-distance is at most sub-exponential in the scale $L$ (see also \cite[Theorem 2.9]{GT}). This makes the MSA unsuitable for tree graphs, where the volume growth is exponential. In this setting, the FMM is applicable \cite{A}.

Although it appears these two methods go on two completely different paths, the developments made for one method can have consequences on the other. In \cite{ETV10, ETV11}, the authors study the case where the dependence of the model on the  parameter $\om$ is non-monotonous, which presents several technical challenges. They were able to extend the FMM to this setting obtaining \eqref{eq:fmm}, however, the existing proofs of dynamical localization derived from this estimate could not be applied to the non-monotonous case. The authors tackled this problem by showing that the fractional moment bound \eqref{eq:fmm} implies the output of the MSA \eqref{outMSA}. Then they could prove dynamical localization as in \cite{vDK}, see \cite{Tphd} for a detailed discussion.

\begin{rem}Recently, A. Elgart and A. Klein developed a multi-scale method that does not involve Green's function, but is a direct analysis of finite-volume eigenfunctions \cite{EK15, EK16}. They retrieve the localization results known for the Anderson model, introducing novel analytical tools. The motivation behind the search for a proof of localization that does not involve Green's function comes from applications to $N$-particle models.
%The same reason is behind J. Imbrie's work \cite{Imbr,Imbr2}, where he develops a multiscale scheme inspired in renormalization.
\end{rem}

\section{Delocalization and the phase diagram}\label{s:deloc}
While dynamical localization is by now well understood, its absence, called \emph{delocalization} has turned out to be by far more elusive. It has been proven in few models, among them, the Anderson model with decaying randomness (see e.g. \cite{KKO, Sim82, DSS, Kis96}); the magnetic Anderson operator in the continuous setting $\textrm{L}^2(\RR^2)$ \cite{GKS1, GKS2}; see also the work \cite{JSBS} on random polymer models, which includes a type of Bernoulli-Anderson model where the random variables are correlated in pairs (dimer model). The most far-reaching developments concerning delocalization are in the setting of tree-graphs. In \cite{AGKL} the authors found a connection between the Anderson model on the Bethe lattice and many-particle systems with inter-particle interaction. This contributed to the interest the Anderson model on trees gathers within the field of random Schr\"odinger operators.
We take most of the material in this section from \cite{Wa12,ASWrev} and \cite[Section 16]{AW}, and we refer the reader to these expository works for more details on the history and ideas behind delocalization (see also \cite{AW13, AW}).

In the remaining of this section we consider the case of the Bethe lattice $\Gamma=\mathbb B$.

%The study of the Anderson model on $\ell^2(\mathbb B)$ was initiated in \cite{ACATh73} , where they found a recursive expression for the Green's function (a type of self-consistency equation)
%and gave arguments for the existence of the Anderson metal-insulator transition:
One can find two lines of arguments in the literature concerning delocalization for the Anderson model on $\ell^2(\mathbb B)$:
\begin{itemize}
\item[i.] Continuity arguments show the stability of the ac spectrum of $-\Delta_\mathbb B$ under the effect of random perturbations.\\
The first proof of delocalization was given by Klein in \cite{Kl94,Kl98}. Under the assumption that the probability distribution $\mu$ has finite second moment, it is shown that for every $0<\abs{E_0}<2\sqrt{K}$, there exists a disorder parameter $\lambda(E_0)>0$ such that
\[ \sup_{ \substack{\epsilon\in(0,1],E\in(-E_0,E_0)\setminus\{0\}\\ \lambda\in(0,\lambda(E_0)]}} \mathbb E\left( \abs{G_{\lambda,\om}(E+i\epsilon;0,0)}^2\right)<\infty.\]
This implies that $H_{\om,\lambda}$ has purely ac spectrum in $[-E_0,E_0]$ almost surely. Klein's proof is based in a renormalization procedure and a fixed-point argument. In \cite{FHS07} the authors give an alternative proof of Klein's result.

In \cite{ASW06}, the authors show the stability of the ac spectral measure for energies in the spectrum (see Theorem 1.1 therein):
\[ \lim_{\lambda\rightarrow 0}\mathbb E\left( \int_I \abs{\Im G_{\lambda,\om}(E+i0;0,0)- \Im G_{0,\om}(E+i0;0,0)}dE\right)=0,\]
where $\Im G_{\lambda,\om}(E+i0;0,0)$ denotes $\lim_{\epsilon\rightarrow 0}\Im G_{\lambda,\om}(E+i\epsilon;0,0)$.
\item[ii.] Resonant delocalization: In \cite{AW13}, M. Aizenman and S. Warzel give a criterion for ac spectrum in terms of the decay rate of Green's function. They exploit the geometric properties of $\mathbb B$ to show that although the Green's function between distant points might decay, the exponential growth of the volume of balls in the tree can contribute enough to compensate for this and give rise to resonances. As a consequence, the Green's function is not square-summable, and by the Simon-Wolff criterion (Theorem \ref{t:sw}), this implies ac spectrum. For a sketch of the proof we refer the reader to the review article by S. Warzel \cite{Wa12} and for details, see \cite{AW13} and \cite[Chapters 15 and 16]{AW}.
\end{itemize}
 Recall that $\mathbb B$ is the graph where every point has $K+1$ neighbors. We fix a point in $\mathbb B$ that we call the root $0$, and we call a point $u\in\mathbb B$ a forward neighbor of $x\in\mathbb B$ if $u\sim x$ and $\dist_{\mathbb B}(u,0)>\dist_{\mathbb B}(x,0)$, that is, $u$ is a neighbor of $x$ away from the direction of the root (see Fig. \ref{f:bet}). We denote by $\mathbb B_0$ the rooted tree at $0$, that is, the tree where all points have $K$ forward neighbors, starting at $0$.

\begin{figure}[h]
  \centering
  \includegraphics[width=3cm]{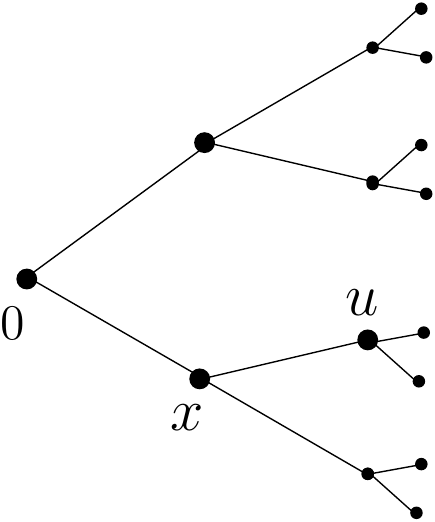}\\
  \caption{In the picture, $u$ is the forward neighbor of $x$ in the graph rooted at $0$.}
  \label{f:bet}
  \end{figure}

The methods to prove delocalization mentioned above rely on a recursion property of the Green's function on the tree, see \cite{ACATh73, Kl98, AW}. We recall a particular case as stated in \cite[Proposition 16.1]{AW},

\begin{thm} Let $\mathbb B_0$ be the rooted tree at $0$ and consider the Anderson model $H_{\omega,\lambda}$ on $\ell^2(\mathbb B_0)$. Then, for any $x\in\mathbb B_0$ and all $z\in\mathbb C\setminus \RR$,
\be G_{\mathbb B_0}(z;0,0)=\left(\lambda V(0)-z-\sum_{y\sim 0}G_{\mathbb B_0}(z;y,y) \right)^{-1}, \ee
and
\be G_{\mathbb B_0}(z;0,x)=\prod_{0\preceq v\preceq x}G_{\mathbb B_0}(z;v,v), \ee
where $0\preceq v\preceq x$ means that $v$ lies in the (unique) path connecting the root $0$ to $x$ in $\mathbb B_0$.
\end{thm}

\begin{rem}The methods mentioned above to prove delocalization have been applied to other models, like tree-strips \cite{Hal12, KlSa, Sha}, decorated trees \cite{FHH12,FHS09}, complete graphs \cite{AShW}, so-called trees of finite cone type \cite{KLW12, KLW15} and random quantum trees \cite{ASW06bis}.
\end{rem}

\subsection{The phase diagram of the Anderson model on $\ell^2(\ZZ^d)$ and $\ell^2(\mathbb B)$}

It is known that in dimension $1$, for any $\lambda>0$ the operator $H_{\omega,\lambda}$ defined in \eqref{eq:and} acting on $\ell^2(\ZZ)$ (i.e. with bounded potential) exhibits localization everywhere in the spectrum. On the other hand, it is conjectured that in $d=2$ the operator still exhibits localization in the whole spectrum, for any value of $\lambda$. However, to this date, for weak disorder the methods only give localization at the edges of the spectrum (recall Theorem \ref{t:andloc}).

 In $d\geq 3$ it is conjectured that there is a transition from localized to delocalized states in the spectrum, which is called the \emph{Anderson metal-insulator transition}. Except for the decaying randomness model and the magnetic Anderson model on $\Lp{\RR^2}$ mentioned above, this remains an open problem. For the Anderson model on $\ell^2(\mathbb B)$ a transition has been proven in the parameter $\lambda$: there exist critical values $\lambda_0,\lambda_*>0$ such that $H_{\omega,\lambda}$ exhibits delocalization in the whole spectrum for $\lambda \in[0,\lambda_0)$, and localization in the whole spectrum for all $\lambda\geq \lambda_*$.

See the figure below for a comparison between the phase diagrams of the Anderson model on $\ell^2(\ZZ^d)$ with $d\geq3$ and $\ell^2(\mathbb B)$ with $K>1$.

%
%\begin{tikzpicture}%[xscale=13,yscale=3.8]
%\draw [thick, fill=pink] (-3.5,4)--(-2,0) to [out=110,in=180](0,2.5) -- (0,2.5) to [out=0,in=70]  (2,0) -- (3.5,4);
%\draw [<->] (0,4.5) -- (0,0) -- (4,0);
%\draw (0,-0.5) -- (0,0) -- (-4,0);
%%\draw[green, domain=-2:2] plot (\x, {(2 -\x)*(2+\x)});
%%\draw[blue, domain=-3.14:3.14] plot (\x, {2*(cos(\x r))});
%\draw [ultra thick, blue] (-2,0)  -- (2,0);
%\draw[dashed, black] (-3,1.5) -- (3,1.5);
%\draw[dashed, black] (-3.5,3.5) -- (3.5,3.5);
%%\draw [thick, red] (-2.2,1.5) circle [x radius=0.5, y radius=0.2];
%%\draw [thick, red] (2.2,1.5) circle [x radius=0.5, y radius=0.2];
%%\draw [blue] (-1,1)  -- (1.5,1);
%%\draw [dotted, blue] (1.5,0) -- (1.5,1);
%%\draw [dotted, blue] (-1,0) -- (-1,1);
%%\draw [thick] (-.1,1) node[left]{1} -- (0.1,1);
%\node at (-2,-0.4) {$-2d$};
%\node at (2,-0.4) {$2d$};
%\node at (4,0.25) {$\sigma(H_{\omega, \lambda})$};
%%\node at (1,2) {Localisation?};
%\node at (1.5,3) {Localization\checkmark};
%\node at (1,1) {ext. states?};
%\node at (0.7,0.6) {($d\geq3$)};
%\node at (0.25,4.5) {$\lambda$};
%%\draw [fill, red] (2.93,2.5) circle [radius=.01];
%%\node [green] at (3,3) {$x\mapsto (2-x)(2+x)$};
%%\node [blue] at (3,1) {$x\mapsto 1_\Ic (x)$};
%\end{tikzpicture}

\begin{figure}[h]
  \centering
  \includegraphics[width=14cm]{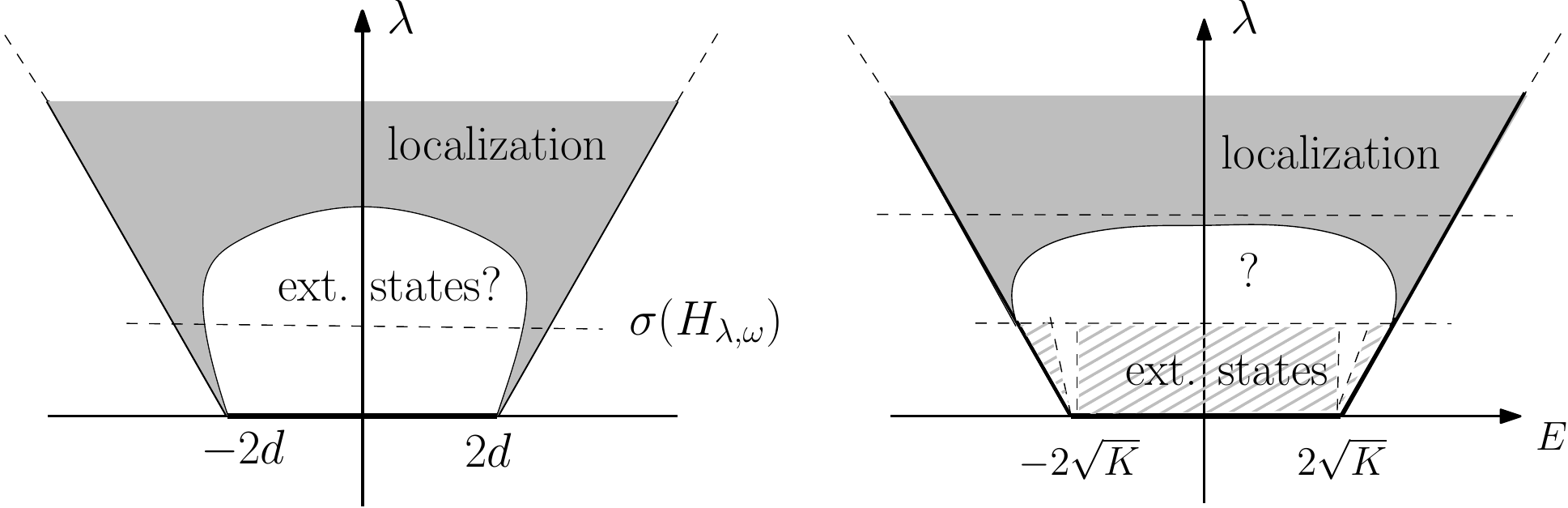}\\
 \caption{The figure on the left shows the phase diagram for the Anderson model $H_{\lambda,\omega}$ on $\ell^2(\mathbb Z^d)$ for $d\geq3$, while the on on the right shows it for the Anderson model on $\ell^2(\mathbb B)$. The horizontal lines represent the spectrum of the operator. Note that on the left, only the localization area is known, while in the figure on the right, only the lower dashed and upper grey sections are known, the white section being a conjecture. The Anderson transition for the Anderson model on $\ell^2(\mathbb B)$ is therefore in the disorder parameter $\lambda$.}
%\label{}
\end{figure}

 The phenomenon of delocalization and the Anderson transition has also been investigated in models other than the Anderson operator on graphs: the Anderson transition has been obtained for the so-called supersymmetric $\sigma$-model \cite{DSZ}, while \cite{ErKYY} gives a proof of eigenfunction delocalization in certain random matrices.

 More recently, N. Anantharaman and M. Sabri have investigated the delocalization of eigenfunctions in the absolutely continuous spectrum for the Anderson model on graphs in the context of quantum ergodicity \cite{AnS1,AnS2}.

%\section{The Integrated Density of States (IDS)}\label{s:ids}

\section*{Acknowledgements}
 The author acknowledges the support of the Hausdorff Center for Mathematics, Bonn and the organizing committee of the CIMPA School in Kairouan in November 2016. The author would like to thank S. Gol\'enia for enjoyable discussions and helpful remarks, and to F. Hoecker-Escuti and A. Klein for comments on the manuscript.

\end{document}